\documentclass[lettersize,journal]{IEEEtran}
\usepackage{cite}
\usepackage{graphicx}
\usepackage{amsmath,amssymb,amsfonts}
\usepackage[justification=centering,font=footnotesize]{caption}
\usepackage{xcolor}
\usepackage{bm}
\usepackage{mathtools}
\usepackage{nccmath}
\usepackage{enumitem}
\usepackage{dsfont}
\usepackage{multicol}
\usepackage{circuitikz}
\usepackage{verbatim}
\usepackage{stfloats}

\usepackage{amsthm}
\newtheorem{theorem}{Theorem}
\newtheorem{lemma}{Lemma}
\newtheorem{proposition}{Proposition}
\newtheorem{definition}{Definition}

\DeclareMathOperator*{\argmin}{arg\,min}

\title{From AoI to QVAoI: Query-Based Semantics-Aware Scheduling for Energy-Harvesting IoT Systems}

\author{
	\IEEEauthorblockN{Erfan Delfani and Nikolaos Pappas, \IEEEmembership{Senior Member, IEEE}} 
	\thanks{The authors are with the Department of Computer and Information Science at Linköping University, Sweden, email: \{\texttt{erfan.delfani, nikolaos.pappas\}@liu.se}. This work has been supported by the Swedish Research Council (VR), ELLIIT, and the European Union (ETHER, 101096526, ELIXIRION, 101120135, 6G-LEADER, 101192080, and SOVEREIGN, 101131481). A shorter version has been published in \cite{delfani2024QVAoI}.}
}

\begin{document}
	
	\maketitle
	
	\begin{abstract}
		In this work, we study the freshness and significance of information in an IoT status update system in which an Energy Harvesting (EH) device samples an information source and forwards update packets to a destination node via a direct channel. We introduce and optimize a semantics-aware metric, Query Version Age of Information (QVAoI), in the system along with other metrics: Query Age of Information (QAoI), Version Age of Information (VAoI), and Age of Information (AoI). We formulate the optimization problem as a Markov Decision Process to determine the optimal transmission policy at the device, which decides the time slots for transmitting updates, subject to the device's battery energy limitations and the energy arrivals. Furthermore, we derive closed-form expressions for the average update rate and the QVAoI for a unit-capacity battery, serving as analytical benchmarks. We compare the performance of QVAoI-Optimal, QAoI-Optimal, VoI-Optimal, and AoI-Optimal policies with a baseline greedy policy. All semantics-aware policies achieve better performance than the greedy policy. The QVAoI-Optimal policy, in particular, demonstrates a significant performance improvement either by providing fresher, more relevant, and more valuable updates with the same energy arrivals or by reducing the number of transmissions in the system while maintaining the same level of freshness and information significance as the QAoI-Optimal and other policies.
	\end{abstract}

	
	\section{Introduction}
	Communicating timely and informative data is crucial for status update systems within real-time IoT networks. These systems involve the sensing and transmission of update packets from an information source through a network to a destination monitor for further processing \cite{yates2021age,abd2019role}. Information packets are sampled using IoT devices and shared within a network, facilitating various applications related to the monitoring and controlling of remote environments, particularly in the contexts of smart cities, intelligent industries, smart agriculture, metaverse, and healthcare, as shown in Fig. \ref{fig_HighLevelSetup}. However, there are limitations in these networks, as IoT devices are highly energy-limited, and the network resources are restricted in terms of bandwidth, channel reliability or availability, and other factors. These limitations necessitate more effective, cost- and energy-efficient approaches for status updating, particularly when communication from a remote, low-energy IoT device (such as an Energy Harvesting (EH) sensor located in an isolated area or a battery-powered wearable device, whether worn or implanted) to a destination monitor is involved.
	
	The \emph{semantics-aware communication} paradigm introduces novel approaches that address the transmission and utilization of the right amount of data at the right time to achieve the designated goals in status update systems\cite{kountouris2021semantics}. This is accomplished by employing semantics-aware performance metrics and managing the information chain from generation to transmission and utilization within the communication network. Recent studies have demonstrated substantial benefits of this paradigm in status update systems \cite{yates2021age}. 

    \begin{figure*}[t!]
        \centering
        \includegraphics[width=5.2in]{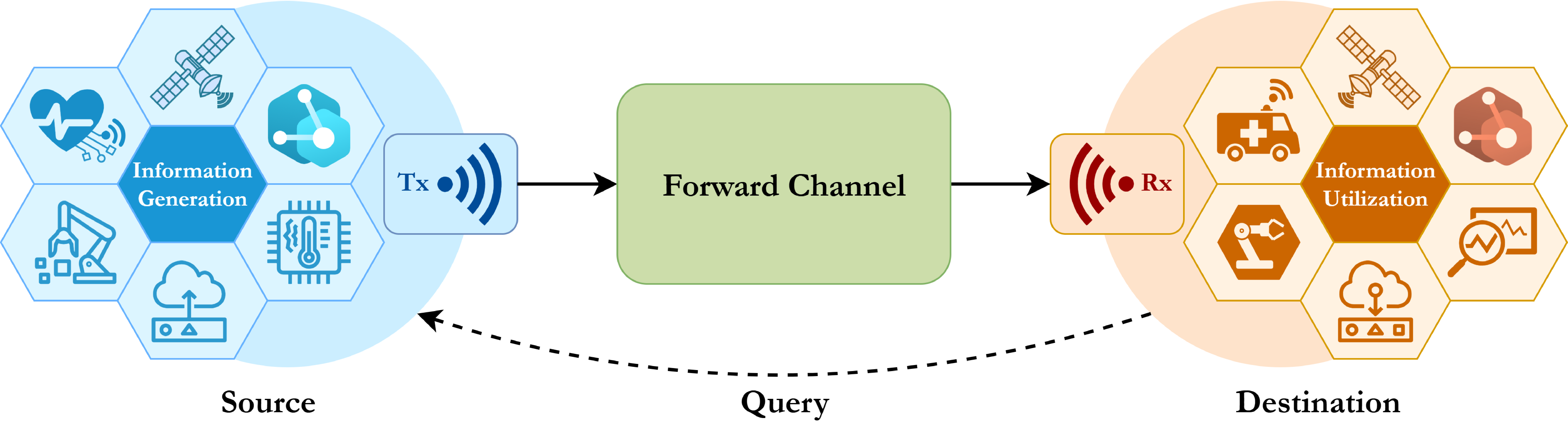}
        \caption{Status update system for various IoT applications. In pull-based setups, the receiver requests the data when needed by sending queries, whereas in push-based setups, the transmitter forwards data to the receiver regardless of the queries.}
        \label{fig_HighLevelSetup}
    \end{figure*}
	
	At the core of semantics-aware communication are semantic metrics that capture the \emph{semantic attributes} of information: \emph{freshness}, \emph{relevance}, and \emph{value}. The freshness of information relates to the staleness of information packets in the network from their generation until their utilization. Higher freshness corresponds to lower staleness of the data packets. The relevance of information relates to sampling the appropriate piece of information from the source. In contrast, the value of information is the benefit to the destination node of receiving the information, relative to the cost of its transmission. We can refer to the relevance and value of information as its \emph{significance}.
	Various metrics have been introduced in the literature, including Age of Information (AoI) \cite{kaul2012real}, non-linear AoI \cite{zheng2019closed}, Age of Incorrect Information (AoII) \cite{maatouk2020age}, Query Age of Information (QAoI) \cite{chiariotti2022query}, Version Age of Information (VAoI) \cite{yates2021Vage}, and state-aware AoI \cite{delfani2024state,luo2026exploiting}. Among these, AoI is a freshness metric that quantifies the time elapsed since the generation time of the last successfully received packet at the destination node. Version AoI (VAoI) is a metric that jointly quantifies the freshness and relevance of information, measuring the number of versions by which the receiver lags behind the source. Query AoI (QAoI) is another performance metric that represents the freshness and value of information. QAoI considers the AoI only when there is a request from the destination node, i.e., only in query instances when the information is assumed to be useful to the receiver.

    However, these existing metrics only partially capture the attributes of information. VAoI accounts for the evolution of the information source through version changes, but it does not consider whether the information is requested by the receiver. In contrast, QAoI captures the timeliness of information at query instances but does not account for the relevance of the transmitted content in terms of version changes. As a result, in systems where both source dynamics and query-based information access are important, these metrics may lead to suboptimal update policies. This highlights the need to move beyond \emph{blind continuous data monitoring} toward a more efficient framework that prioritizes the delivery of \emph{relevant information when it is needed}.

    In this work, we introduce the Query VAoI (QVAoI) metric, which captures all three semantic attributes and is particularly suitable for systems where both source evolution and query-based information demand are critical. The proposed QVAoI metric is especially relevant in \emph{pull-based} and highly resource-constrained IoT systems, where updates are only used when requested and when the underlying information has changed significantly. 
    Such scenarios arise in applications such as remote monitoring with event-driven queries, digital twins, satellite IoT, and industrial IoT control loops, where unnecessary transmissions increase energy consumption without improving system performance. For example, in event-driven IoT monitoring and control systems, the destination node requests relevant fresh data only when a specific event or anomaly is detected. In addition, in satellite IoT, particularly in LEO satellite networks, where connectivity is often intermittent, QVAoI enables query-based updates that reduce unnecessary transmissions and conserve energy. 
    More broadly, QVAoI provides a principled foundation for communication-efficient protocols by enabling systems to prioritize transmissions that are both informative about source evolution and useful to the receiver under query-based access.

    We investigate the advantages of employing and optimizing VAoI and QVAoI within an EH status update system, compared to AoI and its query-based counterpart, QAoI. As content-based metrics, VAoI and QVAoI capture not only freshness but also the relevance and value of information. In contrast, AoI and QAoI do not account for the informativeness or content relevance of updates. We demonstrate that optimizing VAoI and QVAoI in a status update system with an EH IoT device can improve performance and reduce unnecessary transmissions without compromising the conveyed information, particularly for sources that evolve over time.

    \subsection{Related Works}
    The optimization of freshness metrics in EH status update systems has been extensively studied in the literature \cite{yates2021age}. These systems encompass diverse configurations and operate under various time settings. Specifically, in the discrete-time setting, \cite{delfani2024state} considers an EH sensor that monitors a source with normal/alarm macro-states and optimizes a state-aware AoI cost function. The work in \cite{leng2019age} examines an EH secondary sensor in a cognitive radio system, modeling the sequential decision between spectrum sensing and status updating to optimize the average AoI. The authors of \cite{abd2020reinforcement} study an RF-powered IoT device operating over a fading channel, together with its extension to multiple sources, with the objective of minimizing the average or weighted AoI. The study in \cite{ceran2021reinforcement} investigates an EH transmitter communicating over an error-prone channel, which must decide whether to sample, retransmit, or remain idle. Additional scenarios include an EH node with a finite battery receiving updates from heterogeneous sources \cite{gindullina2021age}, which aims to minimize the average AoI; EH remote-sensing systems over time-varying channels \cite{jaiswal2021minimization}, which minimize the time-averaged AoI; and an EH-powered ARQ system \cite{crosara2021stochastic} that minimizes the AoI for successfully delivered packets.
    
    Moreover, studies extend beyond AoI to incorporate relevance and value attributes. There exists a body of work on version-aware \cite{yates2021Vage,buyukates2022version,kaswan2022timely,kaswan2022susceptibility,mitra2022asuman,mitra2024age,kaswan2025age,delfani2023version,mitra2023learning,abd2023distribution,delfani2025LeoSats,salimnejad2025optimizing} and query-based \cite{chiariotti2022query,hatami2022demand,ildiz2023pull,ouguz2022implementation,holm2023goal,agheli2023effective,ayik2023optimization,kriouile2023pull,zakeri2023query,hatami2024status,shiraishi2024coexistence,cosandal2024modeling,kalor2025data,liu2025optimizing} metrics. Specifically, in EH scenarios, several works have been proposed. These include on-demand AoI minimization in multi-user EH IoT networks under various constraints, for both full \cite{hatami2022demand} and partial battery knowledge \cite{hatami2024status}, where on-demand AoI quantifies the freshness of information at the users’ request instants; optimization of the QAoI metric in a pull-based EH sensor model \cite{holm2021freshness,chiariotti2022query}, with QAoI reflecting the freshness at the instants when the receiver needs the data; minimization of VAoI in cache-enabled gossiping networks \cite{delfani2023version} and in LEO satellite networks \cite{delfani2025LeoSats}, which collect versions of data from an EH sensor; optimization of AoII and distortion-based metrics for semantic-aware sampling in tracking systems with partial observability \cite{zakeri2024semantic}; and optimization of Version Innovation Age \cite{salimnejad2025optimizing} for remote tracking of a Markovian source using an EH sensor, considering the state of the source, its importance, in addition to version evolution.

    \subsection{Contributions}
    The main contributions of this paper are as follows:
    \begin{itemize}[leftmargin=1.5em]
        \item We introduce QVAoI, a new semantic metric that captures the timeliness, relevance, and value of information, and optimize it for an EH status update system.
        \item We analytically identify the existence and structure of the optimal update policy for optimizing QVAoI and other metrics, including AoI, VAoI, and QAoI. We prove that the optimal policy adopts a threshold-based structure.
        \item We derive the closed-form average update rate and QVAoI for a system with a minimal unit-capacity (single-size) battery under a threshold policy with an arbitrary threshold $\Delta_T$.
        \item We use simulations to validate the analytical results and to illustrate the impact of system parameters on the optimal average QVAoI and other metrics by comparing them with baseline policies.
    \end{itemize}
	
	\section{Semantic metrics: From AoI to QVAoI}
	\label{SemanticMetrics_QVAoOI}
	In this section, we present the formal definitions of the metrics AoI, QAoI, and VAoI, and introduce our proposed metric, QVAoI. AoI is defined as the time elapsed since the generation time of the last received update, i.e., $\Delta^{AoI}(t) \triangleq t-u(t)$, where $u(t)$ is the timestamp of the current update at the receiver. The AoI is the most widely used metric in the literature, and its optimization yields the freshest updates in status update systems across various configurations \cite{yates2021age,pappas2022age}.
	
	VAoI quantifies the number of versions by which the receiver lags behind the source, i.e., $\Delta^{\mathit{VAoI}}(t) \triangleq N_S(t)-N_R(t)$, where $N_S(t)$ and $N_R(t)$ represent the version numbers of the current updates at the source and receiver, respectively. By version, we refer to any significant change in the content or a new generation of information at the source. 
	
	QAoI is an extension of AoI that considers the age at instances when there is a request from the destination node. This metric is suitable for \emph{pull-based} status update systems, wherein the destination node requests and controls the generation or transmission of updates as needed and when they are valuable for utilization. A pull-based setup contrasts with a \emph{push-based} setup, in which the source or transmitter decides to push updates to the receiver at its discretion, regardless of requests from the receiver. If we presume a binary request arrival process $r(t)$, where $r(t)$ equals $1$ when there is a request from the receiver and $0$ otherwise, then QAoI is defined as $\Delta^{\mathit{QAoI}}(t)  \triangleq  r(t) \times \Delta^{AoI}(t)$. This metric penalizes system staleness only when requests are demanded. In a practical scenario where there is a lag ($\tau$) between the request time and when the system controls are applied, the definition can be corrected to $\Delta^{\mathit{QAoI}}(t)  \triangleq  r(t-\tau) \times \Delta^{AoI}(t)$. 
	
	We introduce the QVAoI metric as an extension of VAoI and QAoI that considers both content changes at the source and queries from the destination. This metric combines the advantages of both VAoI and QAoI to represent relevant and valuable updates alongside the freshest ones. We define the QVAoI as the count of versions by which the receiver lags behind the source at query instances. This metric ensures that the system is not penalized in the absence of requests and when the transmission of updates holds no value or utility for the receiver. The mathematical representation of QVAoI is: 
	\begin{align}
		\Delta^{\mathit{QVAoI}}(t)  \triangleq  r(t-\tau) \times \Delta^{\mathit{VAoI}}(t),
	\end{align}
	where $r(t)$ represents the binary request arrival process, and $\tau$ denotes the time lag between the request time and the instant when the system controls (e.g., update policies) are applied.

    An illustration of the aforementioned metrics is shown in Fig. \ref{fig_SemanticMetrics} for a discrete-time setting $t \in \{0,1,2,\dots\}$, where an update is generated in each time slot. We assume that the service time required for a generated update packet to reach the receiver is one time slot. As shown in the figure, two information packets are delivered at reception instances. At the first reception instance, the update packet generated in the previous time slot is delivered to the receiver, causing the AoI to drop to $1$, indicating that one time slot has passed since the packet was generated. The VAoI becomes $0$ because the source and the receiver now hold the same version of the information. At this moment, both QAoI and QVAoI are zero since the receiver has not issued any queries. Afterward, the AoI increases linearly to $7$ until the next reception instance. During this interval, four new versions (out of $7$ update packets) are generated at the information source, increasing the VAoI to $4$, as the receiver becomes $4$ versions behind. In the same period, three consecutive queries arrive from the receiver. As a result, QAoI and QVAoI become equal to AoI and VAoI, respectively, in the following slots ($\tau=1$). The metrics continue to evolve in this fashion.

    \begin{figure}[t!]
        \centering
        \includegraphics[width=3.3in]{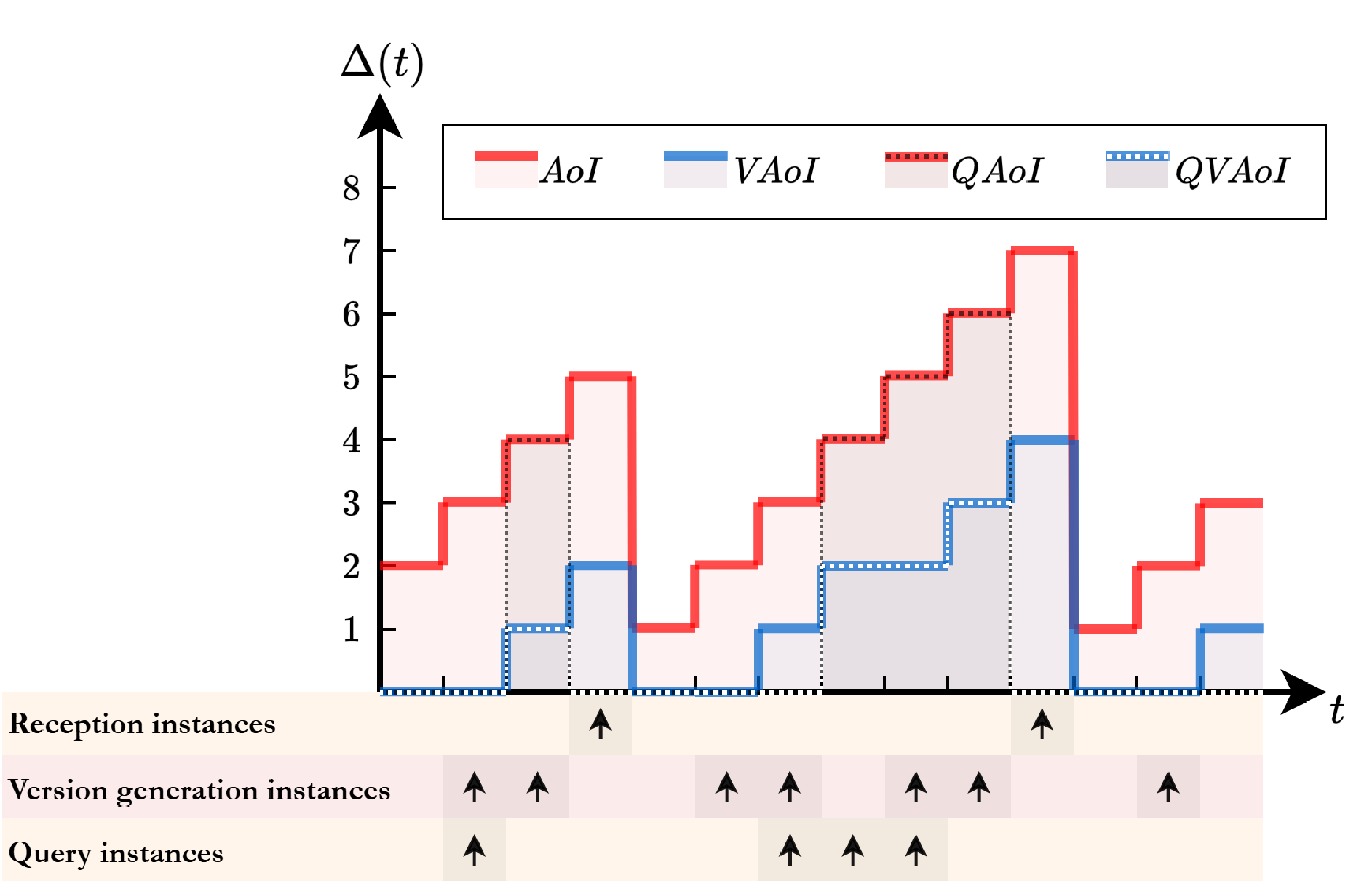}
        \caption{Evolution of metrics over time.}
        \label{fig_SemanticMetrics}
    \end{figure}

    \begin{table}
        \centering
        \caption{Notation of the variables}
        \begin{tabular}{|c|c|}
            \hline
            \textbf{Variable} & \textbf{Description} \\
            \hline
            $s(t)$ & State vector of the system at time slot $t$ \\
            $\Delta(t)$ & Semantic metric at time slot $t$ \\
            $b(t)$ & Battery state at time slot $t$ \\
            $r(t)$ & Query process at time slot $t$ \\
            $q$ & Query arrival probability \\
            $\mathit{\textit{e}}(t)$ & Energy arrival process \\
            $\beta$ & Energy arrival probability \\
            $\mathit{\textit{z}}(t)$ & Version generation process \\
            $p_t$ & Version generation probability \\
            $\mathit{\textit{c}}(t)$ & Channel success process \\
            $p_s$ & Channel success probability \\
            $\pi$ & Update policy \\
            $a(t)$ & Scheduled action at time slot $t$ \\
            $d(t)$ & Realized action at time slot $t$ \\
            $V(s)$ & Value function \\
            $\mu(s)$ & State-stationary probability of state $s$ \\
            $b_{\text{max}}$ & Size of the battery \\
            $\Delta_{\text{max}}$ & Upper bound on the semantic metric \\
            \hline
        \end{tabular}
         \label{tab:placeholder}
    \end{table}
	
	\section{System Model}
    Status update systems provide a fundamental setup for delivering information across a wide range of IoT applications, including real-time healthcare monitoring with wearables, LEO satellite-based observation of remote areas, industrial automation, and digital twin platforms, as illustrated in Fig. \ref{fig_HighLevelSetup}. 
    In this section, we describe the system model for a status update system to optimize metrics AoI and VAoI in a push-based scenario, and QAoI and QVAoI in a pull-based scenario.
    
	\subsection{System Setup}
	We consider the status update system shown in Fig. \ref{SimpleSystemModel_Fig}, in which an EH IoT device samples and transmits update packets to a destination node. The IoT device can be a sensor or a measurement device, and the packets are sampled from an information source, typically a physical process. The destination node can either initiate a request or act as a request aggregator for a destination network, sending queries to the IoT device to obtain new updates.
	
	Our objective is to derive an \emph{update policy} that optimizes the semantic metrics within this system while adhering to the energy constraints imposed on the EH device. This update policy determines the scheduling of optimal times for transmitting new updates, which results in the best performance in terms of metrics. The system operates in discrete time slots, allowing the device to decide whether to transmit an update or remain idle in each slot. We assume that each time slot is long enough to support the sampling and transmission of a new update from the device to the receiver.
	The EH device is equipped with a rechargeable battery of $b_{\text{max}}$ units, which collects energy from ambient sources according to a Bernoulli distribution with probability $\beta$ per time slot. When the battery is empty, the device cannot send updates; otherwise, it must decide when to schedule transmissions. We assume a normalized battery unit so that each transmission consumes one energy unit from the battery. In this system, each update packet from the information source is labeled with a timestamp and a version number, and new versions are generated with probability $p_t$ in each time slot. 
	
	\begin{figure}[bt!]
		\centering
		\includegraphics[width=3.3in,trim={0cm 0cm 0cm 0cm}]{SystemModelSimple.eps}
		\caption{The considered system model.}
		\label{SimpleSystemModel_Fig}
	\end{figure}
	
	We consider an unreliable forward channel from the device to the receiver, where data packets are delivered with success probability $p_s$. We also assume an error-free backward channel from the receiver to the device, used to transmit queries and acknowledgments (ACKs) for successful data packet reception. The device becomes aware of the last timestamp and version stored at the destination node by receiving ACK feedback. We assume a query arrival process $r(t)$ at the receiver, which follows a Bernoulli distribution with probability $q$ per time slot, representing requests from the destination node (or an external network).
	
	\subsection{Problem Formulation}
	
	Our objective is to determine an optimal update policy to optimize the time-averaged expected value of AoI, QAoI, VAoI, and QVAoI, as defined in Section \ref{SemanticMetrics_QVAoOI}.
	\begin{gather}
		\Delta^{AoI}(t)  \triangleq  t-u(t), \\
		\Delta^{\mathit{VAoI}}(t)  \triangleq  N_S(t)-N_R(t), \\
		\Delta^{\mathit{QAoI}}(t)  \triangleq  r(t-1) \times \Delta^{AoI}(t), \\
		\Delta^{\mathit{QVAoI}}(t)  \triangleq  r(t-1) \times \Delta^{\mathit{VAoI}}(t),
	\end{gather}
	where $t$ is the current time, $u(t)$ is the timestamp of the current update at the receiver, $N_S(t)$ and $N_R(t)$ represent the version numbers of the current updates at the source and receiver, respectively, and $r(t)$ is the binary \emph{query arrival process}.

	An update policy, denoted by $\pi$, is a sequence of actions taken by the IoT device over discrete time slots, i.e., $\pi  \triangleq  \big(a^\pi(0), a^\pi(1), a^\pi(2), \cdots \big)$, where $a^\pi(t)$ is the action realized at time $t$ under the policy $\pi$. Specifically, $a^\pi(t) = 1$ indicates a transmission action, while $a^\pi(t) = 0$ denotes an idle action. Due to the stochastic nature of the system variables, the resulting semantic metric under policy $\pi$ is a stochastic process, denoted by $\Delta^{\pi}(t)$, where $\Delta^{\pi}(t)$ can be either $\Delta^{AoI}(t)$, $\Delta^{\mathit{QAoI}}(t)$, $\Delta^{\mathit{VAoI}}(t)$, or $\Delta^{\mathit{QVAoI}}(t)$. The optimization problem can be formulated as follows:
	\begin{align}
		\label{MainOptProb_Eqn}
		\min_{\pi \in \Pi}\ \lim_{T\rightarrow\infty} {\frac{1}{T} E\left[ \sum_{t=0}^{T-1} \Delta^{\pi}(t) \Big| s_0 \right]},
	\end{align}
	where $\Pi$ is the set of all feasible policies and $s_0$ is the system's initial state. \eqref{MainOptProb_Eqn} can be formulated as an infinite-horizon average cost Markov Decision Process (MDP) problem. The MDP problem is characterized by a tuple $<\mathcal{S},\mathcal{A},P,C>$, where $\mathcal{S}$ is the state space, $\mathcal{A}$ is the set of actions, $P$ is the state transition probability function, and $C$ is the cost function.
	
	\begin{itemize}[leftmargin=0.14in]
		\item \textit{States}: The state vector $s(t)$ is defined as $s(t) \triangleq \left[b(t),\Delta(t),r(t)\right]^T \in \mathcal{S}$, where $b(t)$ is the state of the battery, taking value in the set $\mathcal{B} \triangleq \{0,1,2,\ldots,b_{\text{max}}\}$. $\Delta(t) \in \{0,1,2,\cdots,\Delta_{\text{max}}\}$ is either AoI or VAoI at the receiver, and $r(t) \in \{0,1\}$ is the query process at time slot $t$. $r(t)$ is $1$ when there is a query from the receiver and $0$ otherwise. The state space, $\mathcal{S}\! \triangleq \!\big\{(b,\Delta,r)\!:\! \ b \in \mathcal{B},\Delta \in \{1,2,\cdots\!,\Delta_{\text{max}}\} \text{, and } r \in \{0,1\} \big\}$ is a finite set.
		\item \textit{Actions}: At time $t$, $a(t)=0$ represents the action of staying idle, while $a(t)=1$ represents the action of transmitting an update. The action $a(t)$ is forced to be $0$ when $b(t)=0$ or $r(t)=0$. We define the realized action as $d(t) \triangleq \mathds{1}_{\{b(t) \neq 0\}} \times r(t) \times a(t)$, where $\mathds{1}_{\{\cdot\}}$ is the indicator function.
		\item \textit{Transition probabilities}: Given the following equation, 
		\begin{align}
			\label{TransProb_Eqn}
			\mathbb{P}&\left[s(t\!+\!1)| s(t),a(t)\right]\!=\!\mathbb{P}\left[b(t\!+\!1)|b(t),r(t),a(t)\right] \notag \\
			&\times  \mathbb{P}\left[\Delta(t\!+\!1)|\Delta(t),r(t),a(t)\right]\times \mathbb{P}\left[r(t\!+\!1)\right],
		\end{align}
		the transition probabilities are presented in Section \ref{TranProb_Section}.
		\item \textit{Cost function}: The transition cost function equals the metric at the next time slot, where the action affects the state. That is, $C\big( s(t), a(t),s(t+1)\big) \triangleq \Delta(t+1)$, where $\Delta(t)$ can be either $\Delta^{AoI}(t)$, $\Delta^{\mathit{QAoI}}(t)$, $\Delta^{\mathit{VAoI}}(t)$, or $\Delta^{\mathit{QVAoI}}(t)$.
	\end{itemize}

    \textit{Remark:} QVAoI represents a general metric; it reduces to VAoI when $r(t)=1$ for all $t$ (i.e., $q=1$), and it reduces to QAoI when the source generates a new version in every time slot (i.e., $p_t=1$).
	
	\subsection{Transition probabilities}
	\label{TranProb_Section}
	We provide the transition probabilities between the MDP states by introducing the following Bernoulli processes: the \emph{energy arrival process}, $\mathit{\textit{e}}(t)$, the \emph{channel success process}, $\mathit{\textit{c}}(t)$, and the \emph{version generation process}, $\mathit{\textit{z}}(t)$, given by:
	
	\begin{align}
		\label{BernoulliProcesses}
		\begin{matrix}
			\mathit{\textit{e}}(t) \!=\!
			\begin{cases}
				1 & \text{w.p. } \beta, \\
				0 & \text{w.p. } 1\!-\!\beta, \\
			\end{cases} & 
			\mathit{\textit{c}}(t) \!=\! 
			\begin{cases}
				1 & \text{w.p. } p_s, \\
				0 & \text{w.p. } 1\!-\!p_s, \\
			\end{cases} \\
			\mathit{\textit{z}}(t) \!=\! 
			\begin{cases}
				1 & \text{w.p. } p_t, \\
				0 & \text{w.p. } 1\!-\!p_t. \\
			\end{cases}
		\end{matrix}
	\end{align}
	
	We can now characterize the evolution of the states:
	\begin{align}
		b(t\!+\!1) &\!=\! \min \left\{ b(t) \!+\! \mathit{\textit{e}}(t) \!-\! d(t) ,b_{\text{max}} \right\}. \\
		\Delta^{AoI}(t\!+\!1) 
		& \!=\! \begin{cases}
			1, \quad  d(t)\!=\!1 \text{ and } \mathit{\textit{c}}(t)\!=\!1,  \\
			\min \! \left\{ \Delta^{AoI}(t) \!+\! 1 ,\Delta_{\text{max}} \right\}\!, \ \text{o/w}.
		\end{cases} \\
		\Delta^{\!\mathit{VAoI}}(t\!+\!1) 
		& \!=\! \begin{cases}
			\mathit{\textit{z}}(t), \quad  d(t)\!=\!1 \text{ and } \mathit{\textit{c}}(t)\!=\!1, \\
			\min \! \left\{ \Delta^{\!\mathit{VAoI}}(t) \!+\! \mathit{\textit{z}}(t) ,\Delta_{\text{max}} \right\}\!, \ \text{o/w}.
		\end{cases} \\
		r(t\!+\!1) & \!=\! 
		\begin{cases}
			1 & \text{w.p. } q, \\
			0 & \text{w.p. } 1\!-\!q. \\
		\end{cases}
	\end{align}
	
	The transition probabilities can be calculated according to \eqref{TransProb_Eqn} and the following equations:
	\begin{align}
		\mathbb{P}&\left[b(t\!+\!1)\big|b(t),r(t),a(t)\right] \\
        &=
		\begin{cases}
			\beta & d(t)\!=\!0,\  b(t\!+\!1)\!=\!b(t)\!+\!1,\ b(t)\!<\!b_{\text{max}},\\
			\bar{\beta} & d(t)\!=\!0,\  b(t\!+\!1)\!=\!b(t),\ b(t)\!<\!b_{\text{max}}\\
            1 & d(t)\!=\!0,\  b(t\!+\!1)\!=\!b(t)\!=\!b_{\text{max}}\\
			\beta & d(t)\!=\!1,\  b(t\!+\!1)\!=\!b(t),\\
			\bar{\beta} & d(t)\!=\!1,\  b(t\!+\!1)\!=\!b(t)\!-\!1.\\
		\end{cases} \notag
	\end{align}
	
	\begin{align}
		\mathbb{P}&\left[\Delta^{AoI}(t\!+\!1) \big|\Delta^{AoI}(t),r(t),a(t)\right] \\
		&= 
		\begin{cases}
			1 & d(t)\!=\!0,\  \Delta^{AoI}(t\!+\!1)\!=\!\Delta^{AoI}(t)\!+\!1 \!\leq\! \Delta_{\text{max}},\\
            1 & d(t)\!=\!0,\  \Delta^{AoI}(t\!+\!1)\!=\!\Delta^{AoI}(t) \!=\! \Delta_{\text{max}},\\
			\bar{p}_s & d(t)\!=\!1,\  \Delta^{AoI}(t\!+\!1)\!=\!\Delta^{AoI}(t)\!+\!1 \!\leq\! \Delta_{\text{max}},\\
            \bar{p}_s & d(t)\!=\!1,\  \Delta^{AoI}(t\!+\!1)\!=\!\Delta^{AoI}(t)\!=\!\Delta_{\text{max}},\\
			p_s & d(t)\!=\!1,\  \Delta^{AoI}(t\!+\!1)\!=\!1.\\
		\end{cases} \notag
	\end{align}
	
	\begin{align}
		\mathbb{P}&\left[\Delta^{\mathit{VAoI}}(t\!+\!1) \big|\Delta^{\mathit{VAoI}}(t),r(t),a(t)\right] \\
		&= 
		\begin{cases}
			p_t & d(t)\!=\!0,\  \Delta^{\mathit{VAoI}}(t\!+\!1)\!=\!\Delta^{\mathit{VAoI}}(t)\!+\!1 \!\leq\! \Delta_{\text{max}},\\
			\bar{p}_t & d(t)\!=\!0,\  \Delta^{\mathit{VAoI}}(t\!+\!1)\!=\!\Delta^{\mathit{VAoI}}(t)\!\leq\! \Delta_{\text{max}},\\
            1 & d(t)\!=\!0,\  \Delta^{\mathit{VAoI}}(t\!+\!1)\!=\!\Delta^{\mathit{VAoI}}(t)\!=\! \Delta_{\text{max}},\\
			p_t\bar{p}_s & d(t)\!=\!1,\  \Delta^{\mathit{VAoI}}(t\!+\!1)\!=\!\Delta^{\mathit{VAoI}}(t)\!+\!1\!\leq\! \Delta_{\text{max}},\\
			\bar{p}_t\bar{p}_s & d(t)\!=\!1,\  \Delta^{\mathit{VAoI}}(t\!+\!1)\!=\!\Delta^{\mathit{VAoI}}(t)\!\leq\! \Delta_{\text{max}},\\
            \bar{p}_s & d(t)\!=\!1,\  \Delta^{\mathit{VAoI}}(t\!+\!1)\!=\!\Delta^{\mathit{VAoI}}(t)\!=\! \Delta_{\text{max}},\\
			p_tp_s & d(t)\!=\!1,\  \Delta^{\mathit{VAoI}}(t\!+\!1)\!=\!1,\\
			\bar{p}_tp_s & d(t)\!=\!1,\  \Delta^{\mathit{VAoI}}(t\!+\!1)\!=\!0.
		\end{cases} \notag
	\end{align}
	
	\begin{align}
		\mathbb{P}&\left[r(t\!+\!1)\right] = 
		\begin{cases}
			q & r(t\!+\!1)=1,\\
			\bar{q} & r(t\!+\!1)=0,\\
		\end{cases} 
	\end{align}
	where $\bar{\beta} \triangleq 1-\beta$, $\bar{p}_t \triangleq 1-p_t$, $\bar{p}_s \triangleq 1-p_s$, and $\bar{q} \triangleq 1-q$. The total probability theorem can also help simplifying \eqref{TransProb_Eqn}:
	\begin{align}
		P&\!\left[s(t\!+\!1)| s(t),a(t)\right]\! \\
		& \!=\! \sum_{\substack{(z,e,c) \in \{0,1\}^3}} P\left[s(t\!+\!1) \big| s(t),a(t),\mathit{\textit{z}}(t),\mathit{\textit{e}}(t),\mathit{\textit{c}}(t)\right] \! P_c P_e P_z, \notag
	\end{align}
	with $P_c\! \triangleq \!\mathbb{P}\left[\mathit{\textit{c}}(t)\!=\!c\right]$, $P_e\! \triangleq \!\mathbb{P}\left[\mathit{\textit{e}}(t)\!=\!e\right]$, and $P_z\! \triangleq \!\mathbb{P}\left[\mathit{\textit{z}}(t)\!=\!z\right]$ given by \eqref{BernoulliProcesses}.

	\section{Analytical Results} 
	In this section, we discuss the existence and structure of the optimal policies. We proceed with the QVAoI as the cost function since it generally encompasses other metrics.
	
	\begin{definition}
		An MDP is weakly accessible if its states can be divided into two subsets, $S_t$ and $S_c$. States in $S_t$ are transient under any stationary policy, and any state $s^\prime$ in $S_c$ can be reached from any state $s$ in $S_c$ under some stationary policy.
	\end{definition}
	
	\begin{proposition}
		\label{WeaklyAccessProp}
		The MDP problem \eqref{MainOptProb_Eqn} is weakly accessible.
	\end{proposition}
	
	\begin{proof}
		We demonstrate that any state $s^\prime=\left(b^\prime,\Delta^\prime,r^\prime\right) \in \mathcal{S}$ is reachable from any other state $s=\left(b,\Delta,r\right) \in \mathcal{S}$ under a stationary stochastic policy $\pi$, where the action $a\in\{0,1\}$ at each state is randomly selected with positive probability. The state $r^\prime \in \{0,1\}$ is accessible at each state independently of system actions and can remain unchanged with positive probability (w.p.p.). Therefore, we fix $r^\prime$ as states evolve in the remainder of this proof. The state $b^\prime<b$ is reachable from $b$ w.p.p. by executing action $a=1$ for $(b-b^\prime)$ time slots, and $b^\prime \geq b$ is reachable from $b$ w.p.p. by executing action $a=0$ for $(b^\prime-b)$ slots. Upon reaching state $b^\prime$, regardless of subsequent actions, the battery state can remain unchanged w.p.p. Henceforth, we consider the battery state as $b^\prime$ for the rest of the proof. Similarly, the state $\Delta^\prime<\Delta$ can be attained from $\Delta$ w.p.p. by executing action $a=1$ for one time slot, followed by action $a=0$ for $\Delta^\prime$ slots. Conversely, the state $\Delta^\prime \geq \Delta$ is accessible from $\Delta$ by executing action $a=0$ for $(\Delta^\prime-\Delta)$ slots. 
	\end{proof}
	
	\begin{proposition}
		In the MDP problem \eqref{MainOptProb_Eqn}, the optimal average cost $J^\ast$ achieved by an optimal policy $\pi^\ast$ is the same for all initial states, and it satisfies the Bellman’s equation:
		\vspace{-5pt}
		\begin{equation}
			\label{Bellman_eqn}
			J^\ast\!+\!V(s)\!=\!\min_{a\in\left\{0,1\right\}}{\!\bigg\{C(s,a)+\!\sum_{s^\prime\in \mathcal{S}}{\mathbb{P}\left(s^\prime \big | s,a\right)\!V(s^\prime)}\!\bigg\}},
			\vspace{-4pt}
		\end{equation}
		\begin{equation}
			\label{OptimalAction_eqn}
			\pi^\ast(s) \in  \argmin_{a\in\left\{0,1\right\}}{\!\bigg\{C(s,a)+\!\sum_{s^\prime\in \mathcal{S}}{\mathbb{P}\left(s^\prime \big | s,a\right)\!V(s^\prime)}\!\bigg\}},
			\vspace{-4pt}
		\end{equation}
		where $V(s)$ denotes the value function of the MDP problem, $\mathbb{P}\left(s^\prime \big | s,a\right)$ is the transition function, and $C(s,a)$ is the average cost per slot, given by the transition costs: 
		\begin{equation}
			\label{AvgCost_eqn}
			C(s,a) = \sum_{s^\prime\in \mathcal{S}} {\mathbb{P}\left(s^\prime \big | s,a\right) C\left(s,a,s^\prime\right)}, 
		\end{equation}
		with $C(s,a,s^\prime)=r\Delta^\prime$.
	\end{proposition}
	
	\begin{proof}
		According to Proposition \ref{WeaklyAccessProp}, the MDP problem \eqref{MainOptProb_Eqn} is weakly accessible. Consequently, by Proposition 4.2.3 in \cite{bertsekas2011dynamic}, the optimal average cost is invariant across all initial states. Moreover, Proposition 4.2.6 in \cite{bertsekas2011dynamic} guarantees the existence of an optimal policy. According to Proposition 4.2.1 in \cite{bertsekas2011dynamic}, identifying $J^\ast$ and $V(s)$ satisfying \eqref{Bellman_eqn} enables determination of the optimal policy using \eqref{OptimalAction_eqn}.
	\end{proof}
	
	The optimal policy, denoted as $\pi^\ast$, relies on $V(s)$, which typically lacks a closed-form solution. Standard methods, like the (Relative) Value Iteration and Policy Iteration algorithms \cite[Chapter 4]{bertsekas2011dynamic}, can solve this optimization problem.
	
	\begin{definition}
		Suppose there exists a $\Delta_{T}(b) > 0$ for each $b$ such that the action $\pi(b,\Delta,r\!=\!1)$ is $1$ for $\Delta \geq \Delta_{T}(b)$, and $0$ otherwise. In this case, $\pi$ is a threshold policy.
	\end{definition}
	
	\begin{theorem}
		The optimal policy of the MDP problem \eqref{MainOptProb_Eqn} is a threshold policy.
	\end{theorem}
	
	\begin{proof}
		See Appendix \ref{Apen1:Theorem1}.
	\end{proof}

    \begin{figure*}[!ht]
        \centering
        \begin{minipage}{0.32\linewidth}
             \includegraphics[scale=0.44,trim={4.4cm 0cm 0cm 0cm}]{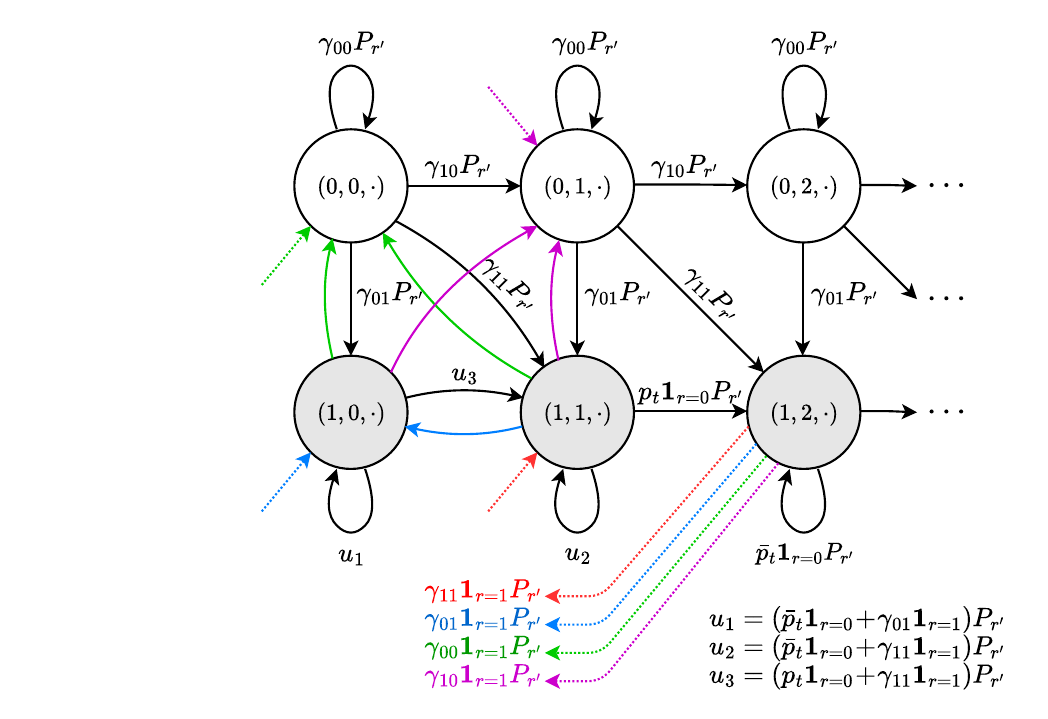}
		\caption{DTMC for $\Delta_T = 0$.}
		\label{fig_DTMC_Bmax1_DT0}
        \end{minipage} \hfill
        \begin{minipage}{0.65\linewidth}
		      \includegraphics[scale=0.44]{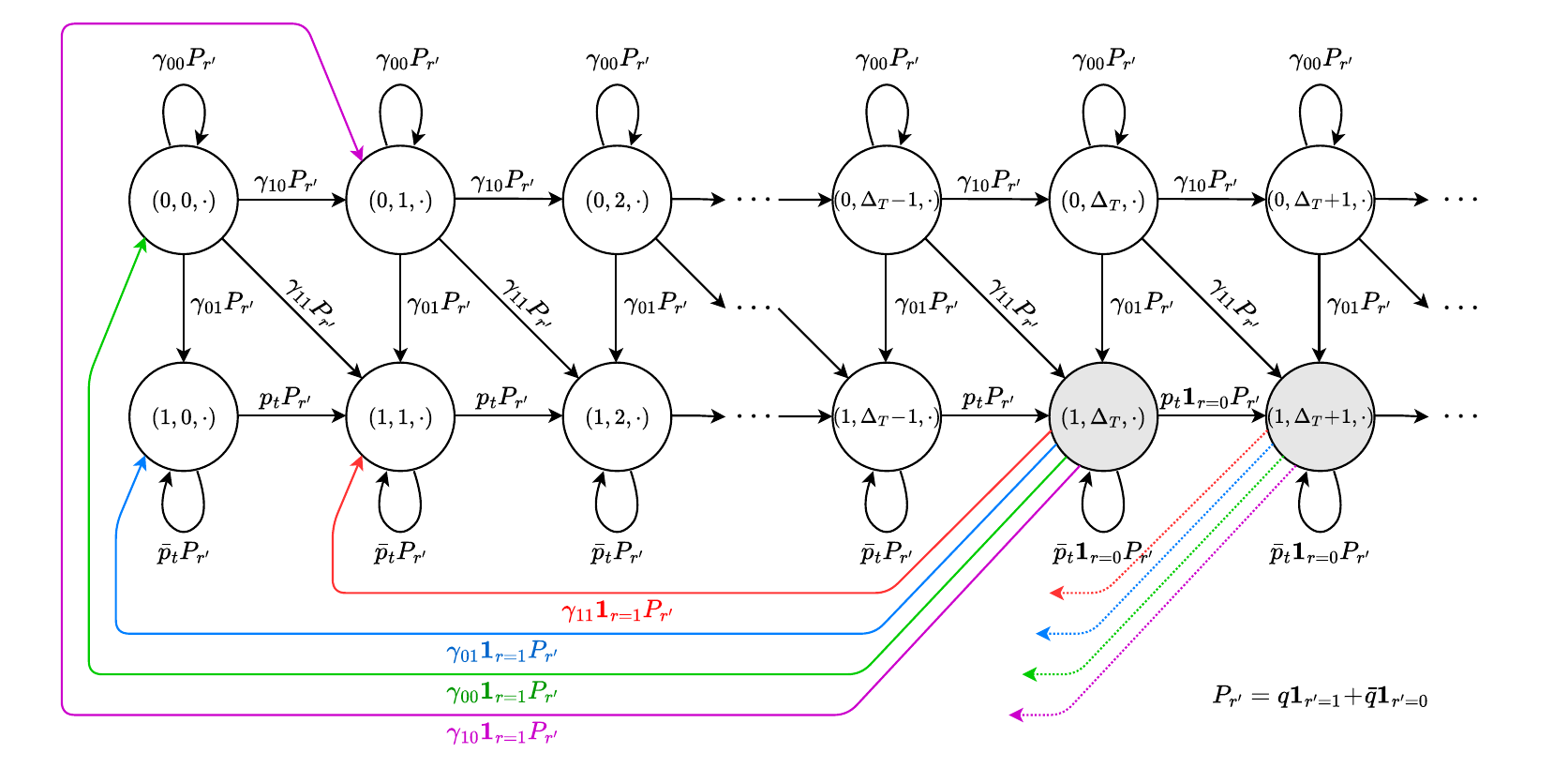}
		      \caption{DTMC for $\Delta_T \geq 1$.}
		      \label{fig_DTMC_Bmax1_DTg1}
        \end{minipage}
    \end{figure*}
	
	\section{Optimal Policy for a Unit-Capacity Battery}
	\label{ThreshPolicies_Bmax1}
	Having established that the optimal policy is a threshold policy, to provide insights into the benefits of semantics-aware information handling, we analyze the average QVAoI of a system with the minimum battery size, i.e., $b_{\text{max}} = 1$. For tractability, we further assume that the channel is reliable ($p_s=1$). To this end, we examine the resulting Markov chains of the system for policies characterized by different threshold values $\Delta_T$, where the action at state $s=(b,\Delta,r)$ is defined as:
		$a(b,\Delta,r) = \mathds{1}_{\{b=1\}} \times \mathds{1}_{\{r=1\}} \times \mathds{1}_{\{\Delta \geq \Delta_T\}}.$
	
    The state space for the Markov chains is defined by $\mathcal{S}_{I} \triangleq \big\{(b,\Delta,r) \big| b\in\{0,1\}, \Delta \in \{0,1,2,\cdots \}, r \in \{0,1\} \big\}$.
	The resulting Markov chains are depicted in Fig. \ref{fig_DTMC_Bmax1_DT0} for $\Delta_T = 0$ and Fig. \ref{fig_DTMC_Bmax1_DTg1} for $\Delta_T \geq 1$. In these figures, for enhanced clarity and compact representation, we define parameters $\gamma_{00} \triangleq \bar{p}_t \bar{\beta}$, $\gamma_{01} \triangleq \bar{p}_t \beta$, $\gamma_{10} \triangleq p_t \bar{\beta}$, and $\gamma_{11} \triangleq p_t \beta$, and we use the same color for the arrows indicating equal transition probabilities. We also use the notation $(b, \Delta, \cdot)$ to represent either $(b, \Delta, r)$ or $(b, \Delta, r')$, where $r$ and $r'$ denote the current and subsequent query states, respectively. The actions $0$ and $1$ are represented by white and grey circles, respectively, at each state $(b,\Delta,1)$. Note that for $(b,\Delta,0)$, the action is always $0$.

    We derive the average update rate, average VAoI, and average QVAoI under a threshold policy by using the stationary distribution of the Markov chains shown in Figs. \ref{fig_DTMC_Bmax1_DT0} and \ref{fig_DTMC_Bmax1_DTg1}. Specifically, the state-stationary probabilities of state $s=(b,\Delta,r)$, denoted by $\mu(s)$, can be obtained by solving the \emph{balance equations}:
	\begin{align}
		\label{BalanceEq_CMDP}
		\bm{\mu} \bm{P_{I}} = \bm{\mu} \text{ and } 
		\sum_{s_i \in \mathcal{S}_{I}} \mu(s_i) = 1,
	\end{align}
	where $\bm{\mu}  \triangleq  \left[ \mu(s_1)\ \mu(s_2)\ \mu(s_3)\ \cdots\ \right]$ is a row vector representing the stationary distributions. $\bm{P_{I}}$ is the transition probability matrix whose $(i,j)^{th}$ element is the transition probability from state $s_i$ to $s_j$, as defined in Section \ref{TranProb_Section} with $b_{\text{max}} = 1$.

    \begin{itemize}
        \item The \emph{average update rate}, denoted as $\Psi$, represents the sum of probabilities over the states where a transmission occurs (i.e., $b=1$, $r=1$, and $\Delta \geq \Delta_T$): 
        \begin{align}
        \label{eqn_AverageRate}
            \Psi = \sum_{\Delta \geq \Delta_T} \mu(1, \Delta, 1).
        \end{align}
	
	   \item The \emph{average VAoI} for these Markov chains is given by:
	   \begin{align}
		\label{eqn_AverageVAoI}
		\Delta^{\mathit{VAoI}}_{Avg} = \sum_{s \in \mathcal{S}_{I}} \Delta \times \mu(s). 
	   \end{align}

        \item The \emph{average QVAoI}, $\Delta_{Avg}^{QVAoI} = \mathbb{E}[r(t-1) \Delta^{VAoI}(t)]$, is obtained by:
        \begin{align}
        \label{eqn_AvgQVAoI}
        \Delta_{Avg}^{QVAoI} = \sum_{s \in \mathcal{S}_{I}} \sum_{s^\prime \in \mathcal{S}_{I}} r \Delta^\prime \times \mathbb{P}(s^\prime \mid s) \times \mu(s),
        \end{align}
        where $s=(b,\Delta,r)$ is the state at time $t-1$ and $s^\prime=(b^\prime,\Delta^\prime,r^\prime)$ is the state at time $t$. 
    \end{itemize}

    \begin{lemma}
        In a setup with a unit-capacity battery, the average update rate for a threshold policy is given by:
        \begin{align}
            \Psi = \frac{q p_t \beta}{q \left(\beta \Delta_T + p_t \rho^{\Delta_T}\right) + \bar{q} p_t \beta},
        \end{align}
        where $\rho  \triangleq  \frac{p_t \bar{\beta}}{1 - \bar{p}_t \bar{\beta}}$.
    \end{lemma}
    \begin{proof}
		See Appendix \ref{Apen_Lemma1}.
	\end{proof}

    \begin{theorem}
    \label{Theorem_AvgQVAoI}
        For a setup with a unit-capacity battery, the average QVAoI and VAoI for a threshold policy are given by \eqref{eqn_AvgQVAoIresults} and \eqref{eqn_AvgVAoIresults}, respectively:

        \begin{align}
            \label{eqn_AvgQVAoIresults}
            \Delta_{Avg}^{QVAoI} \!\!=\! q p_t \!+\! \frac{q \Psi}{p_t \bar{\beta}} \left[ \bar{\beta} \frac{\Delta_T(\Delta_T\!-\!1)}{2} \!+\! \rho^{\Delta_T\!+\!1} \mathcal{G}_{\Delta_T}(\rho) \right]\!,
        \end{align}
        where $\theta  \triangleq  \frac{\bar{q} p_t}{1 - \bar{q}\bar{p}_t}$ and $\mathcal{G}_{\Delta_T}(x)  \triangleq  \frac{\Delta_T - (\Delta_T - 1)x}{(1-x)^2}$.

    \begin{figure*}[!hb]
    \begin{align}
        \label{eqn_AvgVAoIresults}
        \Delta_{Avg}^{VAoI} = 
        \begin{cases}
            \left[ \frac{\rho}{(1-\gamma_{00})(1-\rho)^2} 
            +  \frac{\beta \bar{\beta} \theta}{\bar{q}^2} \frac{1}{\gamma_{10} (1-\gamma_{00}) (\theta-\rho) } \left(\frac{\theta^{2}}{(1-\theta)^2}-\frac{\bar{q}}{\bar{\beta}} \frac{\rho^{2}}{(1-\rho)^2} \right) \right] \Psi,& \Delta_T = 0, \\
            \left[ \frac{\Delta_T(\Delta_T-1)}{2 p_t} + \frac{\rho^{\Delta_T} \mathcal{G}_{\Delta_T}(\rho)}{1-\gamma_{00}} + \frac{\theta \mathcal{G}_{\Delta_T}(\theta)}{\bar{q}} \left( \frac{1}{p_t} - \frac{\rho^{\Delta_T}}{1-\gamma_{00}} \right) + \frac{\beta \rho^{\Delta_T}}{\bar{q} \gamma_{10} (1\!-\!\gamma_{00}) (\frac{1}{\rho} \!-\! \frac{1}{\theta})} \Big( \mathcal{G}_{\Delta_T}(\theta) \!-\! \mathcal{G}_{\Delta_T}(\rho) \Big) \right] \Psi, & \Delta_T \geq 1.
        \end{cases}
    \end{align}
    \end{figure*}
    
    \end{theorem}

    \begin{proof}
		See Appendix \ref{Apen_Theorem_AvgQVAoI}.
	\end{proof}

	\textit{Remark:} In a permanent query setup ($q=1$), the average QVAoI reduces to VAoI; for the thresholds $\Delta_T \in \{0,1\}$, it simplifies to:
		\begin{align}
			\label{AvgQVAoI_B1}
			\Delta^{\mathit{VAoI}}_{Avg}=
			\begin{cases}
				\frac{p_t}{\beta} & \Delta_T=0, \\
				\frac{p_t}{\beta}\left(1-\frac{\bar{p}_t\bar{\beta}\beta^2}{\beta^2+p_t\bar{\beta}(p_t+\beta)}\right) & \Delta_T=1.
			\end{cases}
		\end{align}
	
	In this case, the average VAoI for $\Delta_T=1$ is always less than that for $\Delta_T=0$. However, for $\Delta_T \geq 2$, it may be either lower or higher than that for $\Delta_T=1$, depending on the values of $\Delta_T$, $p_t$, and $\beta$, as further discussed in Section \ref{Sec_NumericalSingle}.

	\section{Numerical Results}
	The MDP problem \eqref{MainOptProb_Eqn} can be solved using the Relative Value Iteration Algorithm (RVIA) to obtain the optimal policies. We evaluate the performance of the AoI-Optimal, VAoI-Optimal, QAoI-Optimal, and QVAoI-Optimal policies with respect to semantic metrics in both pull-based and push-based setups. In pull-based configurations, the QVAoI metric captures the freshness, relevance, and value of the updates within the system, serving as the key performance metric. In push-based setups, the VAoI captures the freshness and relevance of updates within the system, serving as the key performance metric. Given that the optimal policy is a (multi-)threshold policy, we also evaluate the best single-threshold policy for optimizing these metrics (QVAoI and VAoI in pull and push setups, respectively) to determine how closely it approximates the optimal policy. Furthermore, to highlight the advantages of semantics-aware status updating, we include a \emph{greedy policy} as a baseline. Under the greedy policy, the device transmits an update (upon request in a pull-based setup) as soon as energy becomes available.

    In the following simulations, expected values are computed using Monte Carlo methods. We set $\Delta_{\text{max}}$ to a sufficiently high value of $19$ to mitigate the effects of metric truncation while allowing the RVIA algorithm to run effectively over a reasonable state space size.
	
	\subsection{The Impact of Version Generation Probability $p_t$}
	\subsubsection{VAoI-Optimal vs. AoI-Optimal policy in optimizing VAoI}
	The primary advantage of VAoI over AoI is that VAoI considers the dynamics of the information source, leveraging knowledge of the version-generation probability, i.e., $p_t$. In this section, we demonstrate how this advantage can improve the system's performance by maintaining the freshest, most relevant information. We evaluate the average VAoI of the system for the VAoI-Optimal, AoI-Optimal, and greedy policies by varying $p_t$ from $0.2$ to $1$ for $\beta=0.2$, $q=1$, $p_s=1$, and $b_{\text{max}}=15$. As shown in Fig. \ref{fig_VAoIVsPt}, the average VAoI in the system deteriorates as the probability of version generation $p_t$ increases across all three policies. The reason is that the versions of the source become more stale as the energy resources, and thus the possibility of sending updates, remain fixed and unchanged, while new versions are generated with a higher probability. In Fig. \ref{fig_VAoIVsPt}, it is evident that semantics-aware policies, VAoI-Optimal and AoI-Optimal policies, demonstrate superior performance compared to the greedy policy. Additionally, the VAoI-Optimal policy exhibits superior performance compared to the AoI-Optimal policy, especially when $p_t$ is low. 

    \begin{figure}[bt!]
		\centering
		\begin{minipage}{.245\textwidth}
			\includegraphics[width=1.85in,trim={0cm 0cm 0cm 0cm}]{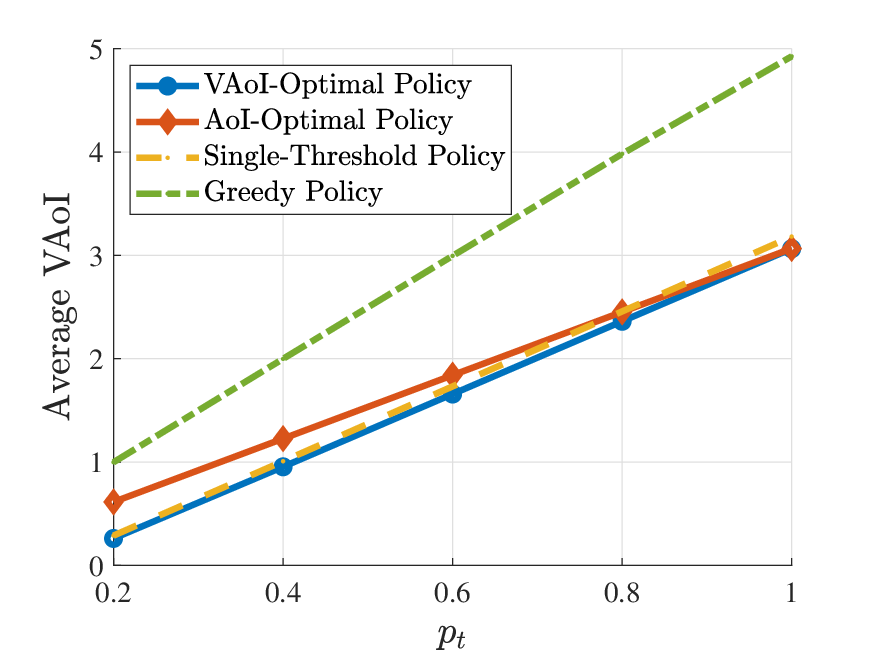}
		      \caption{VAoI vs. $p_t$.}
		      \label{fig_VAoIVsPt}
		\end{minipage}%
		\begin{minipage}{.245\textwidth}
			\includegraphics[width=1.85in,trim={0cm 0cm 0cm 0cm}]{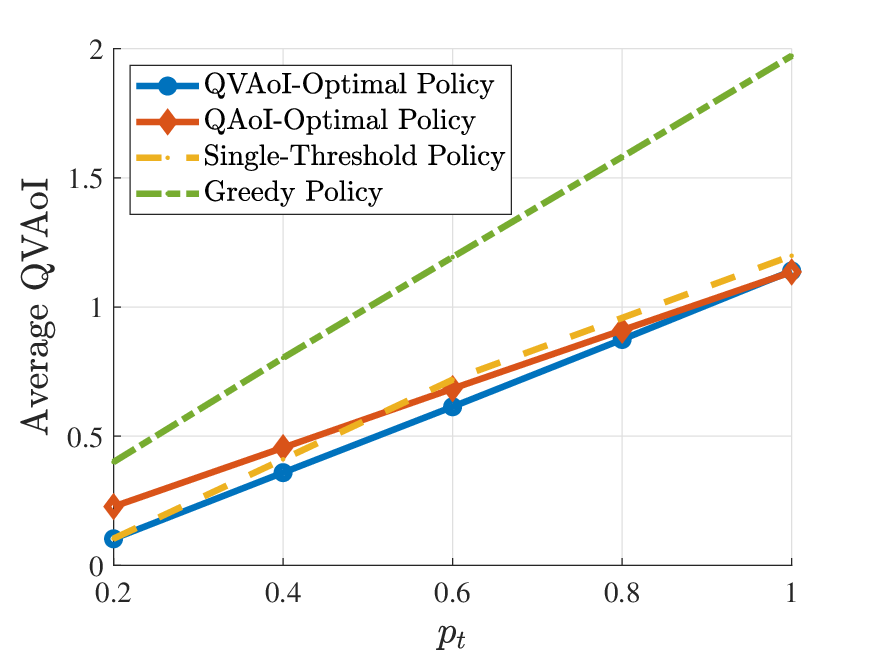}
		      \caption{QVAoI vs. $p_t$.}
		      \label{fig_QVAoIVsPt}
		\end{minipage}%
	\end{figure}

	\subsubsection{QVAoI-Optimal vs. QAoI-Optimal policy in optimizing QVAoI}
	In Fig. \ref{fig_QVAoIVsPt}, the average QVAoI is depicted for QVAoI-Optimal, QAoI-Optimal, and greedy policies as a function of $p_t$, with the parameters $\beta=0.2$, $q=0.5$, $p_s=1$, and $b_{\text{max}}=15$. Here, the semantics-aware policies also demonstrate superior performance compared to the greedy policy, and the QVAoI-Optimal policy outperforms the QAoI-Optimal policy for lower values of $p_t$.
	
	These results emphasize that \emph{incorporating semantics-aware metrics leads to a more effective status updating policy concerning the freshness and significance of information within the system, compared to the greedy policy} (see the red and blue curves vs. the orange curve). Moreover, \emph{the utilization of the VAoI and QVAoI metrics yields enhanced status updating policies concerning the freshness and significance of information, compared to the AoI and QAoI metrics, particularly in scenarios where the source versions evolve slowly, i.e., when $p_t$ is low} (see the blue curve vs. the red curve). It is also evident that the VAoI-Optimal and QVAoI-Optimal policies align with the AoI-Optimal and QAoI-Optimal policies, respectively, when $p_t=1$. 
    In other words, \emph{VAoI and QVAoI represent more general semantic metrics, reducing to AoI and VAoI when the source's version changes in each time slot.}

    \begin{figure}[bt!]
		\centering
		\begin{minipage}{.245\textwidth}
			\includegraphics[width=1.85in,trim={0cm 0cm 0cm 0cm}]{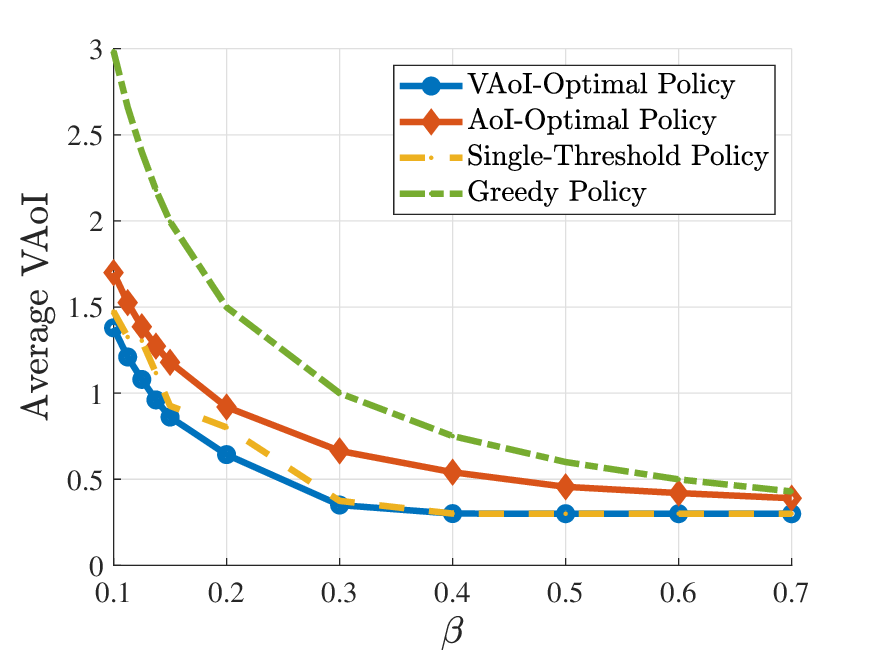}
		      \caption{VAoI vs. $\beta$.}
		      \label{fig_VAoIVsBeta}
		\end{minipage}%
		\begin{minipage}{.245\textwidth}
			\includegraphics[width=1.85in,trim={0cm 0cm 0cm 0cm}]{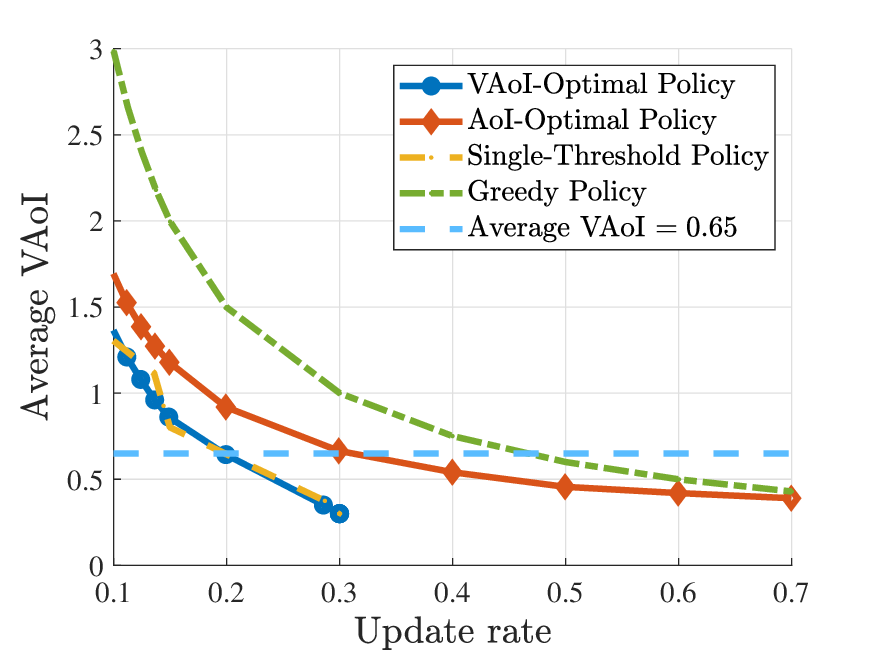}
		      \caption{VAoI vs. update rate.}
		      \label{fig_VAoIVsRate}
		\end{minipage}%
	\end{figure}

	\subsection{The Impact of Energy Arrival Probability $\beta$}
	\subsubsection{VAoI-Optimal vs. AoI-Optimal policy in optimizing VAoI}
	We evaluate the average VAoI of the system across various policies by varying $\beta$ within the range of $0.1$ to $0.7$, while setting the parameters $p_t=0.3$, $q=1$, $p_s=1$, and $b_{\text{max}}=15$. The results are depicted in Fig. \ref{fig_VAoIVsBeta}. First, it can be observed that an increase in the energy arrival probability improves the performance of all three policies, as it provides the device with more energy to send more frequent updates with a higher degree of freedom. For a high energy arrival rate, the three policies perform comparably well. However, as $\beta$ decreases and the system experiences more severe energy limitations, the semantics-aware policies outperform the greedy policy by a considerable margin. Moreover, the VAoI-Optimal policy outperforms the AoI-Optimal policy in optimizing the average VAoI across the system. 
	
	This enhancement becomes more tangible when comparing the average number of updates required by each policy to maintain the same average VAoI. This comparison is illustrated in Fig. \ref{fig_VAoIVsRate}, where the average VAoI is plotted against the update rate, i.e., the average number of updates per slot. In this figure, the horizontal dashed line represents an average VAoI level of $0.65$. The intersection point of this line with the outcomes of the three policies determines the average update rate required for each policy to maintain a performance level of $0.65$ in terms of average VAoI. As shown, the VAoI-Optimal policy requires an update rate of only $0.2$, whereas the AoI-Optimal policy and the greedy policy require update rates of $0.31$ and $0.48$, respectively. This demonstrates improvements of $55\%$ and $140\%$ in average update rate, correspondingly. 
	
	Another result that can be inferred from Fig. \ref{fig_VAoIVsRate} is that, under the VAoI-Optimal policy, the average update rate per slot never exceeds the level of $p_t$ ($0.3$ in the figure). In other words, according to the VAoI-Optimal policy, the update rate per slot is maximized at $p_t$, unless constrained by energy arrivals. 
	
	\subsubsection{QVAoI-Optimal vs. QAoI-Optimal policy in optimizing QVAoI}
	In Figs. \ref{fig_QVAoIVsBeta} and \ref{fig_QVAoIVsRate}, the average QVAoI for QVAoI-Optimal, QAoI-Optimal, and greedy policies is illustrated as a function of $\beta$ and update rate, respectively, with $p_t=0.3$, $q=0.5$, $p_s=1$, and $b_{\text{max}}=15$. In Fig. \ref{fig_QVAoIVsBeta}, we observe that both semantics-aware policies outperform the greedy policy, while the QVAoI-Optimal policy shows superior performance compared to the QAoI-Optimal policy. The average QVAoI is depicted as a function of the update rate in Fig. \ref{fig_QVAoIVsRate}, where the semantics-aware policies require fewer updates to maintain a specific level of average QVAoI. For instance, to achieve an average QVAoI of $0.25$, the QVAoI-Optimal policy results in an update rate of $0.19$, while QAoI-Optimal and greedy policies require $0.28$ and $0.38$, respectively. This demonstrates an improvement of $47\%$ and $100\%$ in the average update rate. In this pull-based scenario, it is evident that the average update rate of the three policies never surpasses the $q$ value. Notably, QVAoI-Optimal exhibits the best performance with a lower update rate.

    \begin{figure}[bt!]
		\centering
		\begin{minipage}{.245\textwidth}
			\includegraphics[width=1.85in,trim={0cm 0cm 0cm 0cm}]{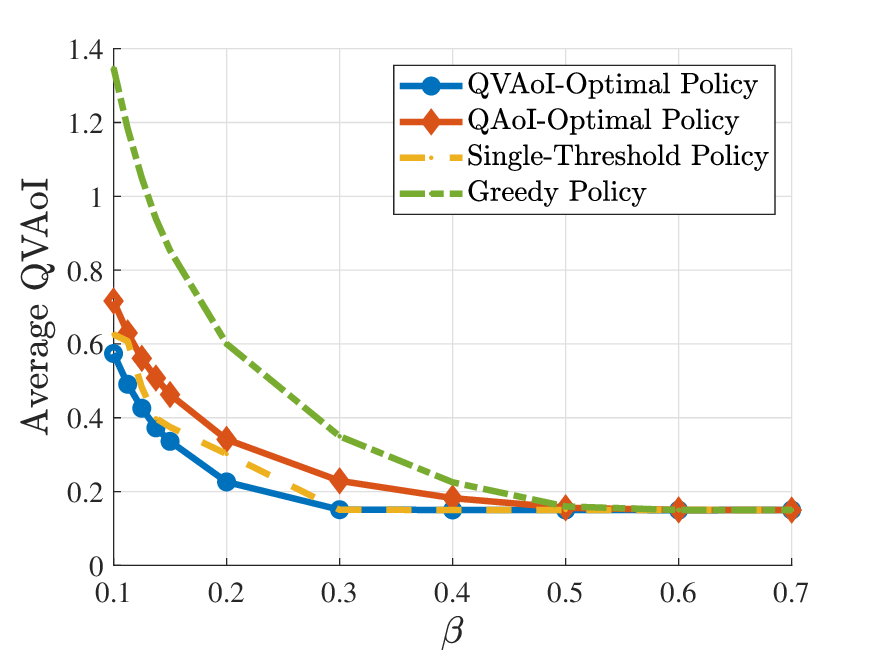}
		      \caption{QVAoI vs. $\beta$.}
		      \label{fig_QVAoIVsBeta}
		\end{minipage}%
		\begin{minipage}{.245\textwidth}
			\includegraphics[width=1.85in,trim={0cm 0cm 0cm 0cm}]{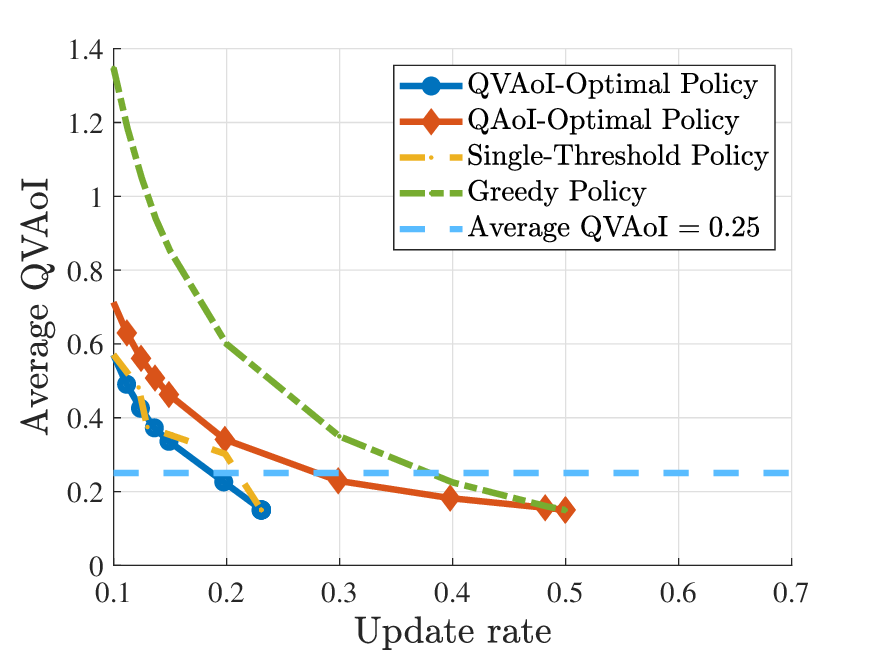}
		      \caption{QVAoI vs. update rate.}
		      \label{fig_QVAoIVsRate}
		\end{minipage}%
	\end{figure}

	These results signify that \emph{in a semantics-aware communication system, we can reduce the number of updates and thus the costs without compromising the conveyed information.} This is particularly important because reducing the number of transmissions in an Energy Harvesting IoT system leads to a significant improvement in energy efficiency. 
	
	\subsection{The Impact of Request Arrival Probability $q$}
	\subsubsection{VAoI-Optimal vs. AoI-Optimal policy in optimizing VAoI}
	The average VAoI of the system as a function of $q$ is depicted in Fig. \ref{fig_VAoIVsQ}, with parameters $\beta=0.2$, $p_t=0.3$, $p_s=1$, and $b_{\text{max}}=15$. In this figure, we observe that $q$ limits the system's performance. For low request arrival probabilities below $\beta$, all three policies perform similarly poorly. However, for high values of $q$, the semantics-aware policies perform better, with the VAoI-Optimal policy demonstrating superior performance compared to the AoI-Optimal policy. This result indicates that a pull-based scenario does not enhance the average VAoI as long as the updates are always valuable to the receiver, or the receiver is always ready to utilize them. In fact, the pull-based scenario is advantageous when there is a limitation on the utilization or value of updates on the receiver side, as discussed in the next section.

    \begin{figure}[bt!]
		\centering
		\begin{minipage}{.245\textwidth}
			\includegraphics[width=1.85in,trim={0cm 0cm 0cm 0cm}]{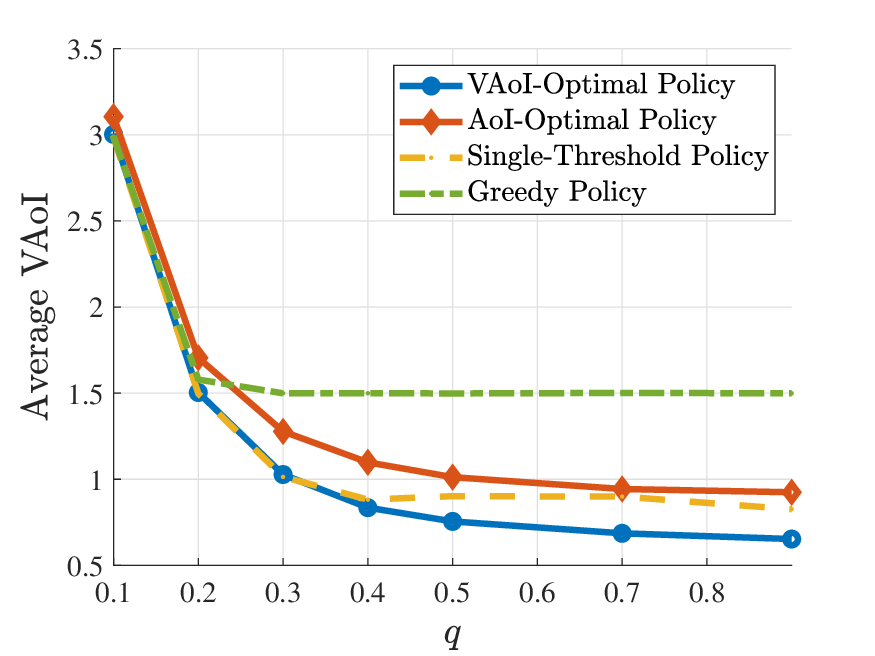}
		      \caption{VAoI vs. $q$.}
		      \label{fig_VAoIVsQ}
		\end{minipage}%
		\begin{minipage}{.245\textwidth}
			\includegraphics[width=1.85in,trim={0cm 0cm 0cm 0cm}]{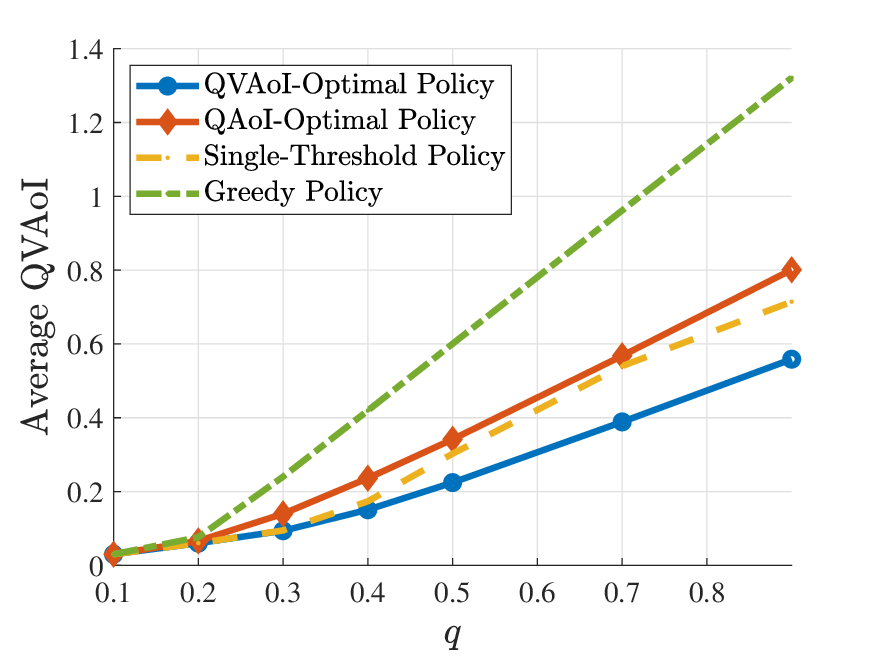}
		      \caption{QVAoI vs. $q$.}
		      \label{fig_QVAoIVsQ}
		\end{minipage}%
	\end{figure}

	\subsubsection{QVAoI-Optimal vs. QAoI-Optimal policy in optimizing QVAoI}
	
	In Fig. \ref{fig_QVAoIVsQ}, the average QVAoI is depicted as a function of $q$. We have fixed the parameters $\beta=0.2$, $p_t=0.3$, $p_s=1$, and $b_{\text{max}}=15$. The average QVAoI increases with $q$ for all three policies, with the QAoI-Optimal policy demonstrating superior performance. 
	The reason behind this increasing behavior is noteworthy. The first reason is that the query process $r(t)$ emerges as a weight in the objective function of the MDP problem \eqref{MainOptProb_Eqn}. As the probability of query arrivals increases, the expected value in this objective function increases. However, even after normalizing the objective function to the expected value of query arrivals, i.e., $\lim_{T\rightarrow\infty} {\frac{1}{T} E\left[ \sum_{t=0}^{T-1} r(t) \right]} = q$, the increasing behavior persists (see Fig. \ref{fig_QVAoIVsQ_norm}). 
	This can be explained by noting that, for higher demands, when the probability of query arrival increases, there are more time slots in which the device must decide whether to transmit. Given the fixed energy arrival rate, the device should set higher thresholds on VAoI values at query instances to limit the number of updates to the maximum permitted by the energy constraint. Consequently, under this policy with higher thresholds, more query instances would result in no action, increasing the average QVAoI. Therefore, increasing pressure on the device and issuing additional requests (beyond $0.3$ in Fig. \ref{fig_QVAoIVsQ_norm}) will degrade the system's average QVAoI. However, concerning the average VAoI, as the probability of query arrivals increases, updates occur frequently enough to reduce the average VAoI, as illustrated in Fig. \ref{fig_VAoIVsQ}.

    \begin{figure}[bt!]
		\centering
		\begin{minipage}{.245\textwidth}
			\includegraphics[width=1.85in,trim={0cm 0cm 0cm 0cm}]{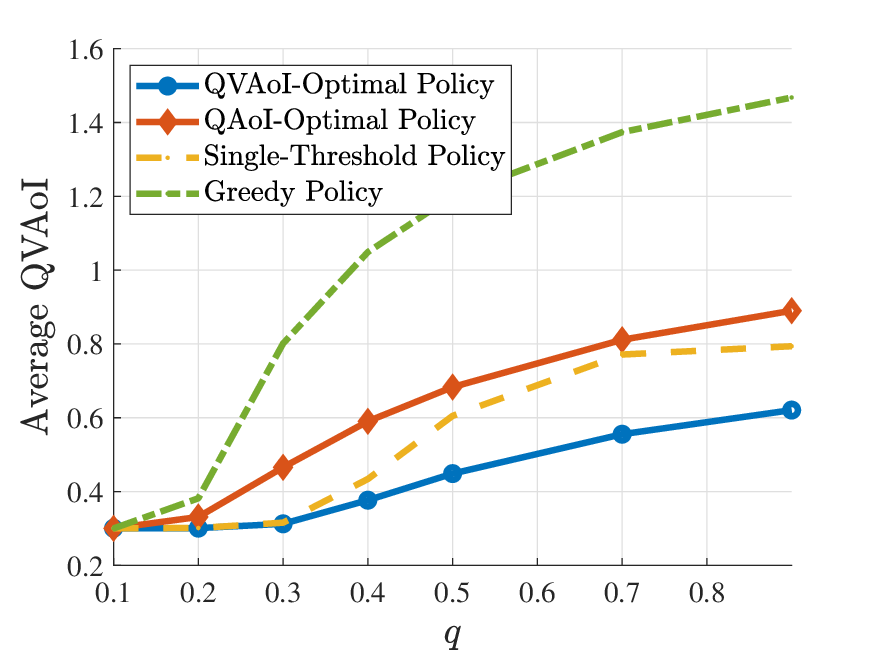}
		\caption{QVAoI (normalized to $q$) vs. $q$.}
		\label{fig_QVAoIVsQ_norm}
		\end{minipage}%
		\begin{minipage}{.245\textwidth}
			\includegraphics[width=1.85in,trim={0cm 0cm 0cm 0cm}]{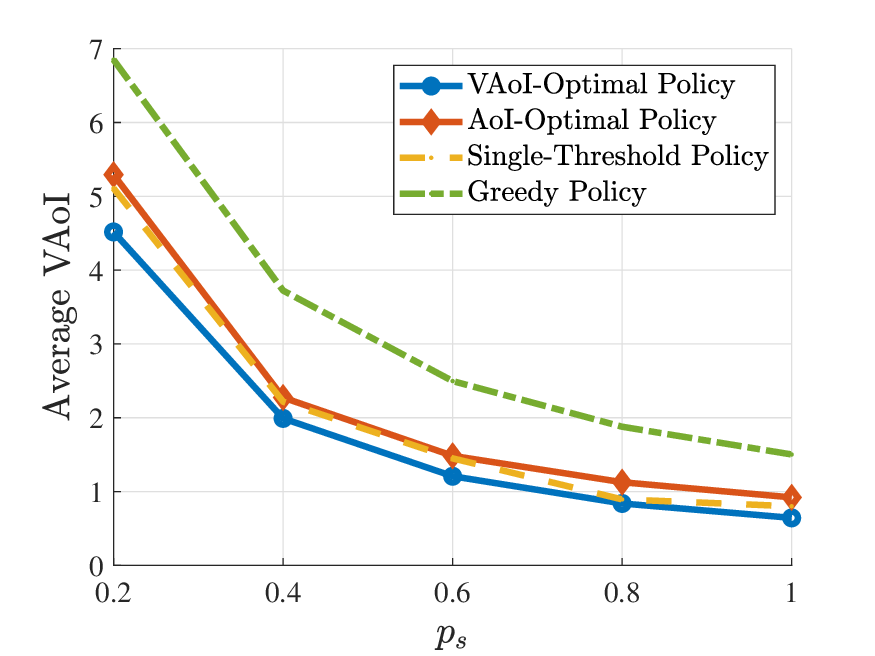}
		\caption{VAoI vs. $p_s$.}
		\label{fig:VAoIVsPs}
		\end{minipage}%
	\end{figure}
	
	\subsection{The Impact of Channel Success Probability $p_s$}
	In Figs. \ref{fig:VAoIVsPs} and \ref{fig:QVAoIVsPs}, the average VAoI and average QVAoI of the system are depicted as functions of $p_s$. We have fixed the parameters $\beta=0.2$, $p_t=0.3$, and $b_{\text{max}}=15$, while $q=1$ and $q=0.5$, respectively. As expected, an increase in channel success probability improves system performance across different policies, with the VAoI-Optimal and QVAoI-Optimal policies achieving the best performance.

    \begin{figure}[bt!]
		\centering
		\begin{minipage}{.245\textwidth}
			\includegraphics[width=1.85in,trim={0cm 0cm 0cm 0cm}]{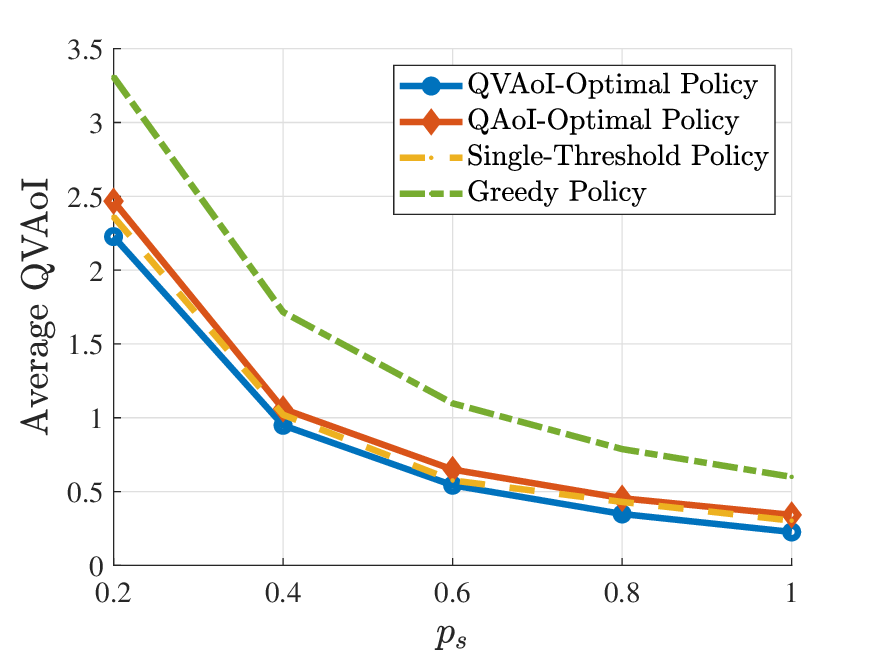}
		\caption{QVAoI vs. $p_s$.}
		\label{fig:QVAoIVsPs}
		\end{minipage}%
		\begin{minipage}{.245\textwidth}
			\includegraphics[width=1.8in,trim={0cm 0cm 0cm 0cm}]{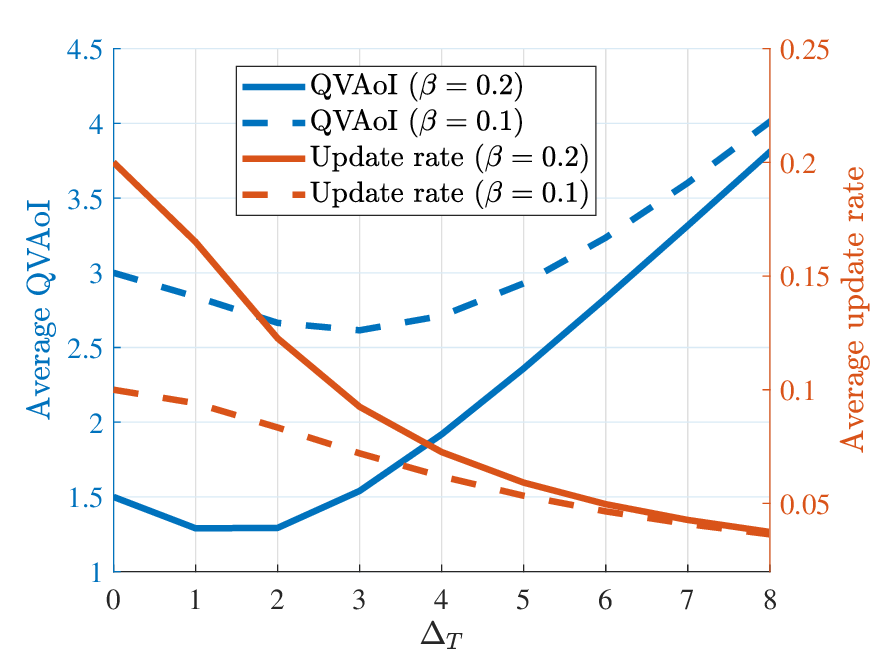}
		  \caption{QVAoI and update rate for unit-capacity battery vs. $\Delta_T$.}
		  \label{fig_QVAoIVsDeltaT}
		\end{minipage}%
	\end{figure}

	\subsection{Average QVAoI for the Unit-Capacity Battery}
    \label{Sec_NumericalSingle}
	We derived the closed-form equations for the average update rate and QVAoI of a threshold policy in a status update system with $b_{\text{max}}\!=\!1$ in Section \ref{ThreshPolicies_Bmax1}. In Fig. \ref{fig_QVAoIVsDeltaT}, the average QVAoI is depicted along the left vertical axis versus different threshold levels $\Delta_T$ for two values of energy arrival probability $\beta$, with $p_t=0.3$ and a high query rate $q=1$. It can be observed that for $\beta=0.2$ (shown with the solid blue curve), the minimum average QVAoI is obtained by setting the threshold $\Delta_T$ equal to $1$. However, by reducing $\beta$ to $0.1$ (the dashed blue curve), the optimal threshold becomes $3$. From this, two conclusions can be drawn. First, the semantics-aware policies outperform the greedy policy (i.e., when $\Delta_T=0$). Second, it is an intriguing finding that, \emph{in scenarios with highly restricted energy arrivals, sending fewer updates results in a fresher system.} This becomes clearer when we compare the update rate per time slot for the optimal policy with a higher threshold to that of a greedy policy with threshold zero, as depicted along the right vertical axis in Fig. \ref{fig_QVAoIVsDeltaT} in red. For instance, when $\beta=0.1$, the update rate corresponding to the optimal policy with a threshold of $3$ is $0.072$, while the greedy policy consumes all the arrived energy and transmits updates as frequently as possible, with an average rate of $0.1$. This occurs while the average QVAoI for the optimal policy is $2.61$ and for the greedy policy is $3$. Thus, \emph{semantics-aware policy results in fewer updates and a fresher system} simultaneously.
    
    Furthermore, for query arrival rates $q \in \{0.1, 0.5, 1\}$, we have depicted the average QVAoI (normalized to $q$) and update rate curves in Fig. \ref{fig_QVAoIVsDeltaT_b02}, for $\beta=0.1$. The average QVAoI for $q=0.1$ is minimized at a threshold of $2$, which is lower than the optimal VAoI threshold ($3$) for $q=1$. This demonstrates that QVAoI, unlike VAoI, restricts updates to query instances. These updates can be triggered at a lower threshold without requiring additional energy (in fact, requiring less) while providing the receiver with updates only when truly needed.

    \begin{figure}[bt!]
        \begin{minipage}{.25\textwidth}
            \includegraphics[width=1.8in,trim={0.5cm 0cm 0cm 0cm}]{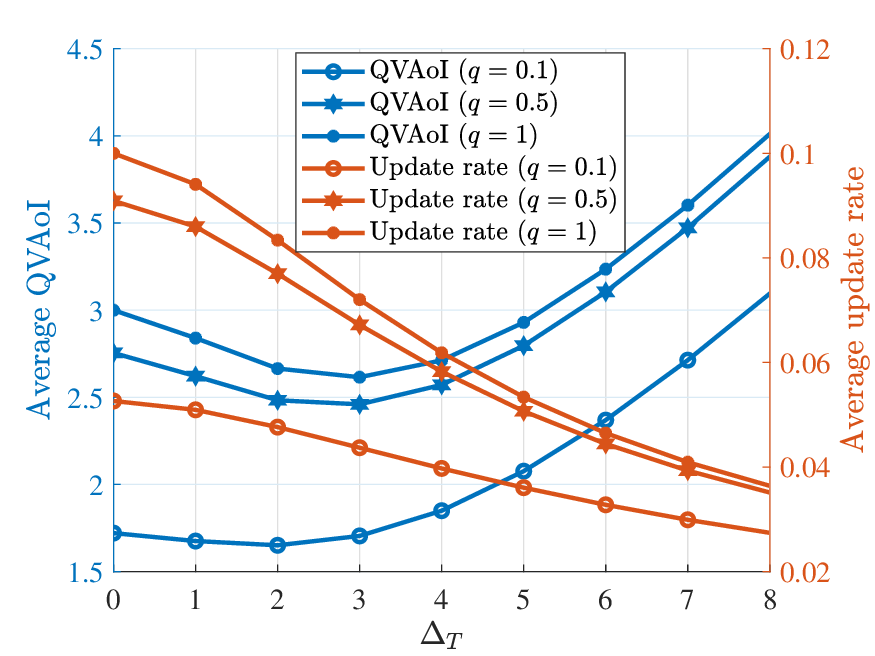}
		      \caption{QVAoI (normalized to $q$) and update rate for unit-capacity battery vs. $\Delta_T$.}
		      \label{fig_QVAoIVsDeltaT_b02}
        \end{minipage}%
		\centering
		\begin{minipage}{.24\textwidth}
			\includegraphics[width=1.8in,trim={0cm 0cm 0cm 0cm}]{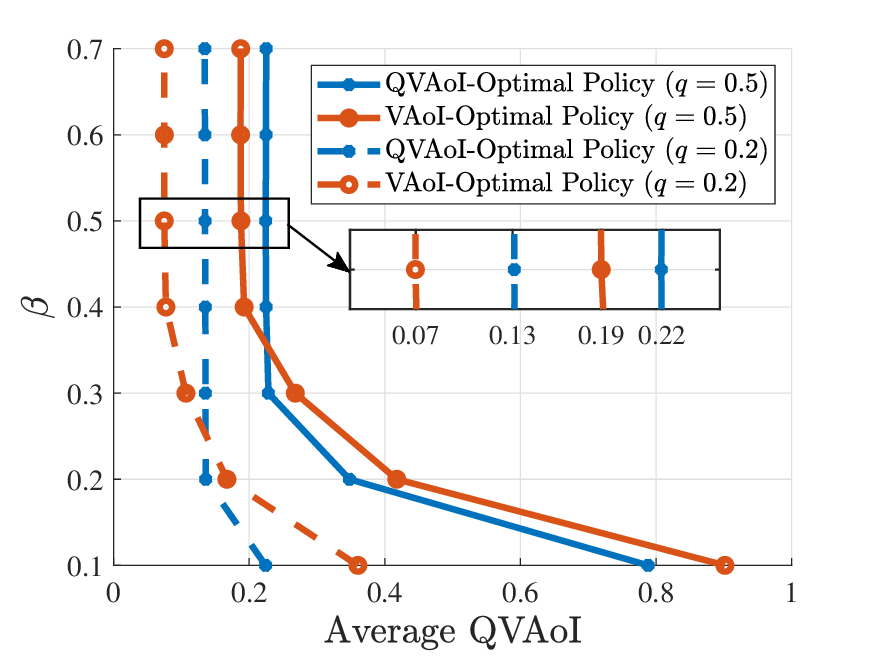}
		\caption{Average QVAoI for QVAoI-Optimal and VAoI-Optimal policies.}
		\label{fig_QVAoI_VAoI_Beta}
		\end{minipage}%
	\end{figure}

    \subsection{Query-Aware vs. Query-Agnostic Policies} 
    As previously discussed, query-aware policies (QAoI- and QVAoI-Optimal) restrict transmissions to query instants, providing updates only when information is valuable within the network. However, a question arises: can proactive transmissions improve system performance in terms of QVAoI? Query-aware policies are most effective in two scenarios: first, when the channel is perfectly reliable ($p_s = 1$), ensuring the delivery of requested information at the precise query instant; and second, in energy-constrained settings (with $p_s < 1$) where the energy arrival rate is lower than the query rate. In the latter case, because energy is scarce, limiting updates to query instants optimizes QVAoI at the receiver. Conversely, when the energy arrival rate is high relative to the query rate, and the channel is unreliable ($p_s < 1$), sufficient energy is available for proactive transmissions outside query instants (in addition to the query instants themselves). In this context, query-agnostic policies (AoI- or VAoI-Optimal) can achieve a lower QVAoI, albeit at the cost of higher energy consumption. This is depicted in Fig. \ref{fig_QVAoI_VAoI_Beta} for parameters $p_t=0.3$, $p_s=0.8$, and $b_{\text{max}}=15$, where the VAoI-Optimal policy curves (red solid and dashed lines for $q=0.5$ and $q=0.2$, respectively) achieve a lower QVAoI than the QVAoI-Optimal policy curves (blue solid and dashed lines for $q=0.5$ and $q=0.2$, respectively) when operating under high energy arrival probabilities.

	\section{Conclusion}
	In this study, we addressed the optimization of freshness and significance of information in a status update system in which an EH device was tasked with scheduling the transmission of measured update packets from an information source to a destination node. We introduced a semantics-aware metric, QVAoI, and identified the QVAoI-Optimal, QAoI-Optimal, VAoI-Optimal, and AoI-Optimal policies by formulating and solving MDP problems. By comparing these with a greedy policy, we demonstrated that the semantics-aware policies delivered superior performance in terms of information freshness and significance. We conducted a closed-form analysis of the unit-capacity battery, offering fundamental insights into system performance under a threshold policy. We also illustrated that the QVAoI-Optimal and VAoI-Optimal policies can achieve fresher and more significant updates from the device, or reduce the number of transmissions without compromising the freshness or significance of information, compared to the QAoI-Optimal and AoI-Optimal policies, respectively. 

	\bibliographystyle{IEEEtran}
	\bibliography{Refs}
	
	\appendices
	
	\section{Proof of Theorem 1}
	\label{Apen1:Theorem1}
	\begin{proof}
		The Bellman equation at state $s\!=\!\left(b,\Delta,r\right)$ is given by:
		
		\begin{align}
			J^\ast\!+\!V(s)\!&=\!\min_{a\in\left\{0,1\right\}} {\bigg\{\underbrace{\sum_{s^\prime\in S} \mathbb{P}\Big[s^\prime \big|s,a\Big] \Big( r\Delta^\prime \!+\! V(s^\prime) \Big)}_{ \triangleq \ Q(s,a)}\bigg\}} \\
			a^\ast(s)&=\argmin_{a\in\left\{0,1\right\}}{Q(s,a)}=
			\begin{cases}
				0, & DV(s)\! \geq \!0,\\
				1, & DV(s)\! < \!0,\\
			\end{cases}
		\end{align}
		where $DV(s)\! \triangleq \!V^1(s)\!-\!V^0(s)$, $V^0(s) \triangleq Q\left(s,0\right)$, and $V^1(s) \triangleq Q\left(s,1\right)$.
		As can be seen, the optimal action $a^\ast(s)$ is related to the sign of $DV(s)$. When $b=0$ or $r=0$, the action $a=0$ is forced and $DV(s)=0$. For other cases where $b>0$ and $r=1$, from Section \ref{TranProb_Section} we have:

		\begin{align}
			\label{eqn_DV_1rate}
			\begin{array}{l}
				V^0(s) \!=\! \sum_{s^\prime\in S} \mathbb{P}\big[s^\prime \big|s,a\!=\!0\big] \Big( \Delta^\prime \!+\! V(s^\prime) \Big) \\
				\!=\! \sum_{r^\prime\in \{0,1\}} \sum_{\substack{z \in \{0,1\} \\ e \in \{0,1\}}} \Big\{ \big(\Delta\!+\!z\big) \!+\! V(b\!+\!e,\Delta\!+\!z,r^\prime) \Big\} P_e P_z P_{r^\prime}, \\
				V^1(s) \!=\! \sum_{s^\prime\in S} \mathbb{P}\big[s^\prime \big|s,a\!=\!1\big] \Big( \Delta^\prime \!+\! V(s^\prime) \Big) \\
				\!=\! \sum_{r^\prime\in \{0,1\}} \sum_{\substack{z \in \{0,1\} \\ e \in \{0,1\}}} \Big\{ \bar{p}_s\big[\left(\Delta\!+\!z\right) \!+\! V(b\!+\!e\!-\!1,\Delta\!+\!z,r^\prime)\big]  \\
				+ p_s \big[ z \!+\! V(b\!+\!e\!-\!1,z,r^\prime) \big] \Big\} P_e P_z P_{r^\prime}.  
			\end{array}
		\end{align}

		In what follows, we demonstrate that $DV(s)=DV(b,\Delta,r)$ is a decreasing (non-increasing) function of $\Delta$, i.e., for $\Delta\!^- \leq \Delta\!^+$, we show that $DV(b,\Delta\!^+,r) \leq DV(b,\Delta\!^-,r)$ or $DV(b,\Delta\!^+,r) - DV(b,\Delta\!^-,r) \leq 0$. This results in the threshold policy because if $DV(s)$ is negative for a $\Delta_T$, it will also be negative for $\Delta \geq \Delta_T$, and the optimal action remains $1$. By simplification of $DV(b,\Delta\!^+,r)$ and $DV(b,\Delta\!^-,r)$ based on \eqref{eqn_DV_1rate} we obtain the following equations:
		
		\begin{align*}
			\begin{array}{l}
				DV(b,\Delta\!^+,r) - DV(b,\Delta\!^-,r) \\
				= V^1(b,\Delta\!^+,r) \!-\! V^1(b,\Delta\!^-,r) \!-\! \big[V^0(b,\Delta\!^+,r) \!-\! V^0(b,\Delta\!^-,r)\big]  \\
				= \sum_{r^\prime, z, e} \bigg\{ \overbrace{p_s\left(\Delta\!^- - \Delta\!^+\right)}^{\leq 0} \\
				\qquad + \bar{p}_s\Big[V(b\!+\!e\!-\!1,\Delta\!^+\!+\!z,r^\prime) - V(b\!+\!e\!-\!1,\Delta\!^-\!+\!z,r^\prime)\Big]  \\
				\qquad - \Big[ V(b\!+\!e,\Delta\!^+\!+\!z,r^\prime) - V(b\!+\!e,\Delta\!^-\!+\!z,r^\prime)\Big] \bigg\} P_e P_z P_{r^\prime} 
			\end{array}
		\end{align*}
		
		Therefore, to verify the inequality $DV(b,\Delta\!^+,r) - DV(b,\Delta\!^-,r) \leq 0$, it is sufficient to show that 
		\begin{multline}
			\bar{p}_s \big[V(b\!-\!1,\Delta\!^+,r) \!-\! V(b\!-\!1,\Delta\!^-,r) \big] \notag \\
			\!-\! \big[ V(b,\Delta\!^+,r) \!-\! V(b,\Delta\!^-,r) \big] \!\leq\! 0,
		\end{multline}
		for $b>0$ and $\Delta\!^- \leq \Delta\!^+$. 
		To proceed with the proof, we use the VIA and mathematical induction. VIA converges to the value function of Bellman's equation regardless of the initial value of $V_0(s)$, i.e., $\lim_{k\rightarrow\infty}{V_k(s)}=V(s),\ \forall s\in S$.
		\begin{align}
			\label{VIA_Iter_k}
			V_{k+1}(s)\!=\!\min_{a\in\left\{0,1\right\}} {\bigg\{\sum_{s^\prime\in S} \mathbb{P}\Big[s^\prime \big|s,a\Big] \Big( r\Delta^\prime \!+\! V_k(s^\prime) \Big)\bigg\}}.
		\end{align}
		
		Therefore, it is sufficient to prove the following inequality for all $k \in \{0,1,2,\cdots\}$:
		\begin{multline}
			\label{eqn_DV_diff_iter_k}
			\bar{p}_s \big[V_k(b\!-\!1,\Delta\!^+,r) \!-\! V_k(b\!-\!1,\Delta\!^-,r) \big] \\
			\!-\! \big[ V_k(b,\Delta\!^+,r) \!-\! V_k(b,\Delta\!^-,r) \big] \leq 0.
		\end{multline}
		
		Assuming $V_0(s) = 0$ for all $s \in S$, equation \eqref{eqn_DV_diff_iter_k} is true for $k=0$. Now, by extending assumption \eqref{eqn_DV_diff_iter_k} for $k > 0$, we aim to prove its validity for $k+1$, i.e.,
		\begin{multline}
			\label{eqn_DV_diff_iter_kp1}
			\bar{p}_s \big[V_{k+1}(b\!-\!1,\Delta\!^+,r) - V_{k+1}(b\!-\!1,\Delta\!^-,r) \big] \\
			- \big[ V_{k+1}(b,\Delta\!^+,r) - V_{k+1}(b,\Delta\!^-,r) \big] \leq 0.
		\end{multline}
		
		the VIA equation \eqref{VIA_Iter_k} is given by $V_{k+1}(s) = \min \{V^0_{k+1}(s),V^1_{k+1}(s)\}$, by defining:
		\begin{align}
			\begin{array}{l}
				V^0_{k+1}(s) \triangleq \sum_{s^\prime\in S} \mathbb{P}\big[s^\prime \big|s,a=0\big] \big( r\Delta^\prime \!+\! V_k(s^\prime) \big), \\
				V^1_{k+1}(s) \triangleq \sum_{s^\prime\in S} \mathbb{P}\big[s^\prime \big|s,a=1\big] \big( r\Delta^\prime \!+\! V_k(s^\prime) \big),
			\end{array}
		\end{align}
		where 
		\begin{align}
			\label{eqn_RVI_VoV1}
			\begin{array}{l}
				V^0_{k+1}(s) \!=\! \sum_{r^\prime, z, e} \Big\{ \big(\Delta\!+\!z\big) \!+\! V_k(b\!+\!e,\Delta\!+\!z,r^\prime) \Big\} P_e P_z P_{r^\prime}, \\
				V^1_{k+1}(s)\!=\! \sum_{r^\prime, z, e} \Big\{ \bar{p}_s\big[\left(\Delta\!+\!z\right) \!+\! V_k(b\!+\!e\!-\!1,\Delta\!+\!z,r^\prime)\big] \\
				+ p_s \big[ z \!+\! V_k(b\!+\!e\!-\!1,z,r^\prime) \big] \Big\} P_e P_z P_{r^\prime} .
			\end{array}
		\end{align}
		
		The inequality \eqref{eqn_DV_diff_iter_kp1} can further be simplified:
		\begin{align}
			\label{eqn_DV_diff_iter_kp1_LastIneq}
			&\bar{p}_s \Big[ \min \{V^0_{k+1}(b\!-\!1,\Delta\!^+,r),V^1_{k+1}(b\!-\!1,\Delta\!^+,r)\} \\
			&\qquad \qquad -\! \min \{V^0_{k+1}(b\!-\!1,\Delta\!^-,r),V^1_{k+1}(b\!-\!1,\Delta\!^-,r)\} \Big] \notag \\ 
			&- \Big[ \min \{V^0_{k+1}(b,\Delta\!^+,r),V^1_{k+1}(b,\Delta\!^+,r)\} \notag \\
			&\qquad \qquad - \min \{V^0_{k+1}(b,\Delta\!^-,r),V^1_{k+1}(b,\Delta\!^-,r)\} \Big] \leq 0 \notag
		\end{align}
		
		We consider four cases to proceed with the proof of \eqref{eqn_DV_diff_iter_kp1_LastIneq}.
		{\small 
			\begin{align*}
				\begin{array}{l}
					\text{\normalsize Case 1. } 
					\begin{cases}
						V^0_{k+1}(b\!-\!1,\Delta\!^-,r) \leq V^1_{k+1}(b\!-\!1,\Delta\!^-,r), \\ 
						V^0_{k+1}(b,\Delta\!^+,r) \leq V^1_{k+1}(b,\Delta\!^+,r).
					\end{cases} \\
					\text{\normalsize Case 2. } 
					\begin{cases}
						V^0_{k+1}(b\!-\!1,\Delta\!^-,r) \leq V^1_{k+1}(b\!-\!1,\Delta\!^-,r), \\ 
						V^0_{k+1}(b,\Delta\!^+,r) > V^1_{k+1}(b,\Delta\!^+,r).
					\end{cases} \\
					\text{\normalsize Case 3. } 
					\begin{cases} 
						V^0_{k+1}(b\!-\!1,\Delta\!^-,r) > V^1_{k+1}(b\!-\!1,\Delta\!^-,r), \\ 
						V^0_{k+1}(b,\Delta\!^+,r) \leq V^1_{k+1}(b,\Delta\!^+,r).
					\end{cases} \\
					\text{\normalsize Case 4. } 
					\begin{cases}
						V^0_{k+1}(b\!-\!1,\Delta\!^-,r) > V^1_{k+1}(b\!-\!1,\Delta\!^-,r), \\
						V^0_{k+1}(b,\Delta\!^+,r) > V^1_{k+1}(b,\Delta\!^+,r).
					\end{cases}
				\end{array}
			\end{align*}
		}
		
		We prove the inequality \eqref{eqn_DV_diff_iter_kp1_LastIneq} for case 1; a similar approach can be utilized to prove the other cases.
		
		\textit{\textbf{Case 1.}} $V^0_{k+1}(b\!-\!1,\Delta\!^-,r) \leq V^1_{k+1}(b\!-\!1,\Delta\!^-,r)$ and $V^0_{k+1}(b,\Delta\!^+,r) \leq V^1_{k+1}(b,\Delta\!^+,r)$. In this case, equation \eqref{eqn_DV_diff_iter_kp1_LastIneq} is further simplified:
		\begin{align}
			&\bar{p}_s \Big[ V^0_{k+1}(b\!-\!1,\Delta\!^+,r) \!-\! V^0_{k+1}(b\!-\!1,\Delta\!^-,r) \Big] \! \notag \\
			&+\! \underbrace{ \bar{p}_s \min \{0,V^1_{k+1}(b\!-\!1,\Delta\!^+,r)\!-\!V^0_{k+1}(b\!-\!1,\Delta\!^+,r)\} }_{\leq 0}  \notag \\
			&- \Big[ V^0_{k+1}(b,\Delta\!^+,r) \!-\! V^0_{k+1}(b,\Delta\!^-,r)\Big] \notag \\
			&+ \underbrace{ \min \{0,V^1_{k+1}(b,\Delta\!^-,r)\!-\!V^0_{k+1}(b,\Delta\!^-,r)\} }_{\leq 0}  \leq 0,
		\end{align}
		where we have used $\min\left\{x,y\right\}=x+\min\left\{0,y-x\right\}$. The second and last terms are negative (non-positive), thus it suffices to show that:
		\begin{multline}
			\bar{p}_s \Big[ V^0_{k+1}(b\!-\!1,\Delta\!^+,r) - V^0_{k+1}(b\!-\!1,\Delta\!^-,r) \Big] \\
			- \Big[ V^0_{k+1}(b,\Delta\!^+,r) - V^0_{k+1}(b,\Delta\!^-,r)\Big]    \leq 0.
		\end{multline}
		
		According to \eqref{eqn_RVI_VoV1}, it can be written as follows:
		\begin{align}
			\begin{array}{l}
				\bar{p}_s \sum_{r^\prime, z, e} \Big\{ \big(\Delta\!^+\!-\!\Delta\!^-\big) \!+\! V_k(b\!+\!e\!-1,\Delta\!^+\!+\!z,r^\prime) \notag \\
				\qquad - V_k(b\!+\!e\!-1,\Delta\!^-\!+\!z,r^\prime) \Big\} P_e P_z P_{r^\prime} \notag \\
				- \sum_{r^\prime, z, e} \Big\{ \big(\Delta\!^+\!-\!\Delta\!^-\big) \!+\! V_k(b\!+\!e\!,\Delta\!^+\!+\!z,r^\prime) \notag \\
				\qquad - V_k(b\!+\!e\!,\Delta\!^-\!+\!z,r^\prime) \Big\} P_e P_z P_{r^\prime} \leq 0 \notag \\
				\Leftrightarrow \sum_{r^\prime, z, e} \Big\{ \overbrace{(1-\bar{p}_s)(\Delta\!^--\Delta\!^+)}^{\leq 0} \notag \\
				+ \bar{p}_s \big[ V_k(b\!+\!e\!-1,\Delta\!^+\!+\!z,r^\prime) - V_k(b\!+\!e\!-1,\Delta\!^-\!+\!z,r^\prime) \big] \notag \\
				- \big[ V_k(b\!+\!e\!,\Delta\!^+\!+\!z,r^\prime) \!-\! V_k(b\!+\!e\!,\Delta\!^-\!+\!z,r^\prime) \big] \Big\} P_e P_z P_{r^\prime} \leq 0 \notag
			\end{array}
		\end{align}
		where the first term in the summation is negative since $\Delta\!^- \leq \Delta\!^+$ and $1-\bar{p}_s=p_s > 0$. The remaining terms are also negative according to the induction assumption \eqref{eqn_DV_diff_iter_k}, and the proof is complete.
	\end{proof}

    \section{Proof of Lemma 1}
	\label{Apen_Lemma1}
    \begin{proof}
        The balance equations for the Markov chain shown in Figs. \ref{fig_DTMC_Bmax1_DT0} and \ref{fig_DTMC_Bmax1_DTg1} can be written directly for each state. By solving the balance equations, we obtain the steady-state probability of each state as a function of $\Psi$ as follows:
            \begin{align*}
				&\mathcal{E}_1\!:\! \begin{cases}
					\mu(0,0,0)\!=\!\bar{q} \left[\gamma_{00} \mu(0,0,:)\!+\!\gamma_{00} \Psi \right], \\
                    \mu(0,0,1)\!=q \left[\gamma_{00} \mu(0,0,:)\!+\!\gamma_{00} \Psi \right], 
                \end{cases} \\
				&\mathcal{E}_2\!:\! \begin{cases}
                    \mu(0,1,0)\!=\!\bar{q} \left[ \gamma_{00}\mu(0,1,:)\!+\!\gamma_{10}\mu(0,0,:)\!+\!\gamma_{10} \Psi \right], \\ 
                    \mu(0,1,1)\!=\!q \left[ \gamma_{00}\mu(0,1,:)\!+\!\gamma_{10}\mu(0,0,:)\!+\!\gamma_{10} \Psi \right],
                \end{cases} \\
                &\mathcal{E}_3\!:\! \begin{cases}
					\mu(0,\Delta,0)\!=\!\bar{q} \left[ \gamma_{00}\mu(0,\Delta,:) \!+\! \gamma_{10}\mu(0,\Delta\!-\!1,:) \right], \ \ \Delta \!\geq\! 2, \\
                    \mu(0,\Delta,1)\!=\!q \left[ \gamma_{00}\mu(0,\Delta,:) \!+\! \gamma_{10}\mu(0,\Delta\!-\!1,:) \right], \ \ \Delta \!\geq\! 2,
                \end{cases} \\
                &\mathcal{E}_4\!:\! \begin{cases}
					\mu(1,0,0)\!=\! \bar{q} \left[ \gamma_{01}\mu(0,0,:)\!+\!\bar{p}_t\mu(1,0,:)\!+\!\gamma_{01}\Psi \right], \\
                    \mu(1,0,1)\!=\! q \left[ \gamma_{01}\mu(0,0,:)\!+\!\bar{p}_t\mu(1,0,:)\!+\!\gamma_{01}\Psi \right],
                \end{cases} \\
                &\mathcal{E}_5\!:\! \begin{cases}
					\mu(1,1,0)\!=\bar{q} \big[ \!\gamma_{11}\mu(0,0,:)\!+\!\gamma_{01}\mu(0,1,:)+\!p_t\mu(1,0,:) \\
					\qquad \qquad +\bar{p}_t\mu(1,1,:) \!+\!\gamma_{11} \Psi \big], \\
                    \mu(1,1,1)\!=q \big[ \!\gamma_{11}\mu(0,0,:)\!+\!\gamma_{01}\mu(0,1,:)+\!p_t\mu(1,0,:) \\
					\qquad \qquad +\bar{p}_t\mu(1,1,:) \!+\!\gamma_{11} \Psi \big],
                \end{cases} \\
                &\mathcal{E}_6\!:\! \begin{cases}
					\mu(1,\Delta,0)\!=\!\bar{q} \big[\gamma_{11}\mu(0,\Delta\!-\!1,:)\!+\!\gamma_{01}\mu(0,\Delta,:) \\
					\qquad \qquad + p_t\mu(1,\Delta\!-\!1,:)\!+\!\bar{p}_t\mu(1,\Delta,:) \big], \ \ 2 \!\leq\! \Delta \!<\! \Delta_T, \\
                    \mu(1,\Delta,1)\!=\!q \big[\gamma_{11}\mu(0,\Delta\!-\!1,:)\!+\!\gamma_{01}\mu(0,\Delta,:) \\
					\qquad \qquad + p_t\mu(1,\Delta\!-\!1,:)\!+\!\bar{p}_t\mu(1,\Delta,:) \big], \ \ 2 \!\leq\! \Delta \!<\! \Delta_T,
                \end{cases} \\
                &\mathcal{E}_7\!:\! \begin{cases}
					\mu(1,\Delta_T,0)\!=\! \bar{q} \big[ \gamma_{11}\mu(0,\Delta_T\!-\!1,:)\!+\!\gamma_{01}\mu(0,\Delta_T,:) \\
					\qquad \qquad + p_t\mu(1,\Delta_T\!-\!1,:) \!+\!  \bar{p}_t\mu(1,\Delta_T,0) \big], \\
                    \mu(1,\Delta_T,1)\!=\! q \big[ \gamma_{11}\mu(0,\Delta_T\!-\!1,:)\!+\!\gamma_{01}\mu(0,\Delta_T,:) \\
					\qquad \qquad + p_t\mu(1,\Delta_T\!-\!1,:) \!+\!  \bar{p}_t\mu(1,\Delta_T,0) \big],
                \end{cases} \\ 
                &\mathcal{E}_8\!:\! \begin{cases}
					\mu(1,\Delta,0)\!=\! \bar{q} \big[ \gamma_{11}\mu(0,\Delta\!-\!1,:) \!+\! \gamma_{01}\mu(0,\Delta,:) \\
                    \qquad \qquad + p_t\mu(1,\Delta\!-\!1,0) \!+\!  \bar{p}_t\mu(1,\Delta,0) \big],\ \ \Delta \!>\! \Delta_T, \\
                    \mu(1,\Delta,1)\!=\! q \big[ \gamma_{11}\mu(0,\Delta\!-\!1,:) \!+\! \gamma_{01}\mu(0,\Delta,:) \\
                    \qquad \qquad + p_t\mu(1,\Delta\!-\!1,0) \!+\!  \bar{p}_t\mu(1,\Delta,0) \big],\ \ \Delta \!>\! \Delta_T.
				\end{cases}
			\end{align*}
        
        For compactness, we have used the notation $\mu(b,\Delta,:)$ to represent the sum $\mu(b,\Delta,0) + \mu(b,\Delta,1)$. The equations above are valid for $\Delta_T \geq 2$, though $\mathcal{E}_6$ may be omitted when $\Delta_T = 2$. For the case where $\Delta_T = 1$, equations $\mathcal{E}_1$ through $\mathcal{E}_4$ and $\mathcal{E}_8$ remain valid, $\mathcal{E}_6$ and $\mathcal{E}_7$ are omitted, and $\mathcal{E}_5$ is modified as follows:
        \begin{align*}
            &\mathcal{E}_5^\prime\!:\! \begin{cases}
					\mu(1,1,0)\!=\bar{q} \big[ \!\gamma_{11}\mu(0,0,:)\!+\!\gamma_{01}\mu(0,1,:)+\!p_t\mu(1,0,:) \\
					\qquad \qquad +\bar{p}_t\mu(1,1,0) \!+\!\gamma_{11} \Psi \big], \\
                    \mu(1,1,1)\!=q \big[ \!\gamma_{11}\mu(0,0,:)\!+\!\gamma_{01}\mu(0,1,:)+\!p_t\mu(1,0,:) \\
					\qquad \qquad +\bar{p}_t\mu(1,1,0) \!+\!\gamma_{11} \Psi \big],
                \end{cases}
        \end{align*}
        
        For $\Delta_T = 0$, equations $\mathcal{E}_1$ through $\mathcal{E}_3$ remain valid, $\mathcal{E}_8$ remains valid for $\Delta \geq 2$, $\mathcal{E}_6$ and $\mathcal{E}_7$ are omitted, and $\mathcal{E}_4$ and $\mathcal{E}_5$ are modified as follows:
        \begin{align*}
            &\mathcal{E}''_4\!:\! \begin{cases}
					\mu(1,0,0)\!=\! \bar{q} \left[ \gamma_{01}\mu(0,0,:)\!+\!\bar{p}_t\mu(1,0,0)\!+\!\gamma_{01}\Psi \right], \\
                    \mu(1,0,1)\!=\! q \left[ \gamma_{01}\mu(0,0,:)\!+\!\bar{p}_t\mu(1,0,0)\!+\!\gamma_{01}\Psi \right],
                \end{cases} \\
            &\mathcal{E}''_5\!:\! \begin{cases}
					\mu(1,1,0)\!=\bar{q} \big[ \!\gamma_{11}\mu(0,0,:)\!+\!\gamma_{01}\mu(0,1,:)+\!p_t\mu(1,0,0) \\
					\qquad \qquad +\bar{p}_t\mu(1,1,0) \!+\!\gamma_{11} \Psi \big], \\
                    \mu(1,1,1)\!=q \big[ \!\gamma_{11}\mu(0,0,:)\!+\!\gamma_{01}\mu(0,1,:)+\!p_t\mu(1,0,0) \\
					\qquad \qquad +\bar{p}_t\mu(1,1,0) \!+\!\gamma_{11} \Psi \big],
                \end{cases}
        \end{align*}

        We solve the equation pairs $\mathcal{E}_1$ through $\mathcal{E}_8$ sequentially to obtain the stationary probability of the states as a function of $\Psi$. Subsequently, we solve the following equation to derive $\Psi$.
        \begin{align}
            \label{eqn_SumOneBmax1}
            \sum_{s_i \in \mathcal{S}_{I}} \mu(s_i) = 1.
        \end{align}
        
        Note that, in the equation pairs $\mathcal{E}_1$-$\mathcal{E}_8$, we have $\mu(b,\Delta,0)=\frac{\bar{q}}{q} \mu(b,\Delta,1)$; therefore:
        \begin{align}
            \sum_{\Delta=\Delta_T}^{\infty} \mu(1,\Delta,0) = \frac{\bar{q}}{q} \sum_{\Delta=\Delta_T}^{\infty} \mu(1,\Delta,1) = \frac{\bar{q}}{q} \Psi.
        \end{align}
        or equivalently:
        \begin{align}
            \label{eqn_SumDTtoInfb1}
            \sum_{\Delta=\Delta_T}^{\infty} \mu(1,\Delta,:) = \left(1+\frac{\bar{q}}{q}\right)\Psi = \frac{1}{q} \Psi.
        \end{align}
        We use \eqref{eqn_SumDTtoInfb1} to simplify \eqref{eqn_SumOneBmax1} as follows for $\Delta_T \geq 2$:
        \begin{gather}
                \sum_{\Delta=0}^{\infty}  \mu(0,\Delta,:) \!+\! \sum_{\Delta=0}^{\Delta_T-1} \mu(1,\Delta,:) \!+\! \sum_{\Delta=\Delta_T}^{\infty} \mu(1,\Delta,:) = 1 \notag \\
                \Rightarrow \sum_{\Delta=0}^{\infty}  \mu(0,\Delta,:) + \sum_{\Delta=0}^{\Delta_T-1} \mu(1,\Delta,:) + \frac{1}{q} \Psi = 1.
                \label{eqn_PsiEqnDTg2}
        \end{gather}

        Similarly, for $\Delta_T \in \{0,1\}$, we have:
        \begin{gather}
                \sum_{\Delta=0}^{\infty}  \mu(0,\Delta,:) + \mu(1,0,:) + \frac{1}{q} \Psi = 1, \label{eqn_PsiEqnDT1} \\
                \sum_{\Delta=0}^{\infty}  \mu(0,\Delta,:) + \frac{1}{q} \Psi = 1. \label{eqn_PsiEqnDT0}
        \end{gather}
        We solve equations $\mathcal{E}_1$ through $\mathcal{E}_6$ to facilitate the solution of \eqref{eqn_PsiEqnDTg2}--\eqref{eqn_PsiEqnDT0} and derive $\Psi$.

        Solving the system of two equations in $\mathcal{E}_1$, we obtain $\mu(0,0,0)$ and $\mu(0,0,1)$ as functions of $\Psi$. Specifically, $\mu(0,0,0)$ is derived by substituting $\mu(0,0,1) = \frac{q}{\bar{q}} \mu(0,0,0)$ from the second equation into the first. (This substitution, $\mu(b,\Delta,1) = \frac{q}{\bar{q}} \mu(b,\Delta,0)$, is also applied to equation pairs $\mathcal{E}_2$ through $\mathcal{E}_8$). After simplification, the result for $\mu(0,0,0)$ is substituted back into the expression for $\mu(0,0,1)$ to .:
        \begin{align}
            \mu(0,0,0) = \frac{\bar{q} \gamma_{00}}{1 - \gamma_{00}} \Psi, \quad
            \mu(0,0,1) = \frac{q \gamma_{00}}{1 - \gamma_{00}} \Psi,
        \end{align}

        Similarly, by solving the system of equations in $\mathcal{E}_2$, we obtain $\mu(0,1,0)$ and $\mu(0,1,1)$ as functions of $\Psi$:
        \begin{align}
            \label{eqn_Mu010} 
            \mu(0,1,0) = \frac{\bar{q} \gamma_{10}}{\left(1 - \gamma_{00}\right)^2} \Psi = \frac{\bar{q} \rho}{1 - \gamma_{00}} \Psi, \\
            \label{eqn_Mu011} 
            \mu(0,1,1) = \frac{q \gamma_{10}}{\left(1 - \gamma_{00}\right)^2} \Psi = \frac{q \rho}{1 - \gamma_{00}} \Psi,
        \end{align}
        where we have defined the parameter $\rho$ as:
        \begin{align}
            \rho  \triangleq  \frac{\gamma_{10}}{1 - \gamma_{00}} = \frac{p_t \bar{\beta}}{1 - \bar{p}_t \bar{\beta}}.
        \end{align} 

       To solve the equations in $\mathcal{E}_3$, we define the function $\mathcal{X}(\Delta)$ for $\Delta \geq 1$ as:
       \begin{align}
        \mathcal{X}(\Delta)\! \triangleq \!\mu(0, \Delta, :)\!=\!\mu(0, \Delta, 0) \!+\! \mu(0, \Delta, 1), 
       \end{align}
       where $\mathcal{X}(1)=\frac{\rho}{1 - \gamma_{00}} \Psi$ is obtained from \eqref{eqn_Mu010} and \eqref{eqn_Mu011}. For $\Delta \geq 2$, from $\mathcal{E}_3$, we have $\mathcal{X}(\Delta) = \rho \mathcal{X}(\Delta-1)$; therefore:
        \begin{align}
            \mathcal{X}(\Delta) &= \frac{ \rho^\Delta}{1 - \gamma_{00}}\Psi, \quad \Delta \geq 1, \\
            \mu(0,\Delta,0) &= \bar{q} \rho \mathcal{X}(\Delta-1) = \bar{q} \mathcal{X}(\Delta),  \quad \Delta \geq 1,\\
            \mu(0,\Delta,1) &= q \rho \mathcal{X}(\Delta-1) = q \mathcal{X}(\Delta), \quad \Delta \geq 1.
        \end{align}

        Solving the system of two equations in $\mathcal{E}_4$, we obtain $\mu(1,0,0)$ and $\mu(1,0,1)$ as functions of $\Psi$:
        \begin{align}
            \mu(1,0,0) \!=\! \frac{\bar{q} \gamma_{01}}{p_t \left(1 \!-\! \gamma_{00} \right)} \Psi, \quad 
            \mu(1,0,1) \!=\! \frac{q \gamma_{01}}{p_t \left(1 \!-\! \gamma_{00} \right)} \Psi.
        \end{align}

        Solving the system of two equations in $\mathcal{E}_5$, we obtain $\mu(1,1,0)$ and $\mu(1,1,1)$ as functions of $\Psi$:
        \begin{align}
            \mu(1,1,0) = \frac{\bar{q}}{p_t \left(1 - \gamma_{00} \right)} \left( \gamma_{11} + \frac{\gamma_{01} \gamma_{10}}{1-\gamma_{00}} + \gamma_{01}  \right) \Psi, \\
            \mu(1,1,1) = \frac{q}{p_t \left(1 - \gamma_{00} \right)} \left( \gamma_{11} + \frac{\gamma_{01} \gamma_{10}}{1-\gamma_{00}} + \gamma_{01}  \right) \Psi,
        \end{align}
        which can be simplified further as:
        \begin{align}
            \label{eqn_Mu110}
            \mu(1,1,0) = \bar{q} \left[ \frac{1}{p_t} - \frac{\rho}{1-\gamma_{00}} \right]\Psi, \\
            \label{eqn_Mu111}
            \mu(1,1,1) = q \left[ \frac{1}{p_t} - \frac{\rho}{1-\gamma_{00}} \right]\Psi,
        \end{align}

        To solve the equations in $\mathcal{E}_6$, we define the function $\mathcal{Y}(\Delta)$ for $1 \leq \Delta < \Delta_T$ as:
       \begin{align}
        \mathcal{Y}(\Delta)\! \triangleq \!\mu(1, \Delta, :)\!=\!\mu(1, \Delta, 0) \!+\! \mu(1, \Delta, 1), 
       \end{align}
       where $\mathcal{Y}(1)=\frac{\Psi}{p_t} - \mathcal{X}(1)$ is obtained from \eqref{eqn_Mu110} and \eqref{eqn_Mu111}. For $2 \leq \Delta < \Delta_T$, based on $\mathcal{E}_6$ and the relation $\gamma_{11} \mathcal{X}(\Delta\!-\!1) \!+\! \gamma_{01} \mathcal{X}(\Delta) \!=\! \frac{\gamma_{11}}{1-\gamma_{00}} \mathcal{X}(\Delta\!-\!1) $, we have:
       
       \begin{align}
       \mathcal{Y}(\Delta) = \mathcal{Y}(\Delta-1) + \frac{\beta}{1-\gamma_{00}} \mathcal{X}(\Delta-1).
       \end{align} 

       Solving recursively yields:
       \begin{align}
           \mathcal{Y}(\Delta) &= \mathcal{Y}(1) + \frac{\beta}{1-\gamma_{00}} \sum_{\delta=1}^{\Delta-1} \mathcal{X}(\Delta) \notag \\
           &= \frac{\Psi}{p_t} - \mathcal{X}(1) + \frac{\beta}{1-\gamma_{00}} \frac{1-\rho^{\Delta-1}}{1-\rho} \mathcal{X}(1)  
       \end{align}
       Given that $1-\rho=\frac{\beta}{1-\gamma_{00}}$, this can be further simplified as follows:
       \begin{align}
           \mathcal{Y}(\Delta) = \frac{\Psi}{p_t} - \mathcal{X}(\Delta), \quad 1 \leq \Delta < \Delta_T.  
       \end{align}
       which gives:
       \begin{align}
            \mu(1,\Delta,0) = \bar{q} \mathcal{Y}(\Delta), \quad
            \mu(1,\Delta,1) = q \mathcal{Y}(\Delta), \quad 1 \leq \Delta < \Delta_T.
        \end{align}

        Solving the system of two equations in $\mathcal{E}_7$, we obtain $\mu(1,\Delta_T,0)$ and $\mu(1,\Delta_T,1)$ as functions of $\Psi$:

        \begin{align*}
            \mu(1,\Delta_T,0) &\!=\! \frac{\bar{q} \left[ \gamma_{11} \mathcal{X}(\Delta_T\!-\!1) \!+\!  \gamma_{01} \mathcal{X}(\Delta_T) \!+\! p_t \mathcal{Y}(\Delta_T\!-\!1) \right]}{1-\bar{q} \bar{p}_t}, \\
            \mu(1,\Delta_T,1) &\!=\! \frac{q \left[ \gamma_{11} \mathcal{X}(\Delta_T\!-\!1) \!+\!  \gamma_{01} \mathcal{X}(\Delta_T) \!+\! p_t \mathcal{Y}(\Delta_T\!-\!1) \right]}{1-\bar{q} \bar{p}_t},
        \end{align*}
        which can be further simplified using the functions $\mathcal{X}(\cdot)$ and $\mathcal{Y}(\cdot)$ as follows:
        \begin{align}
            \mu(1, \Delta_T, 0) = \theta \mathcal{Y}(\Delta_T), \quad
            \mu(1, \Delta_T, 1) = \frac{q}{\bar{q}} \theta \mathcal{Y}(\Delta_T),
        \end{align}
        where we have defined the parameter $\theta$ as: $\theta  \triangleq  \frac{\bar{q} p_t}{1 - \bar{q}\bar{p}_t}$.

        Solving the system of two equations in $\mathcal{E}_8$, we obtain $\mu(1,\Delta,0)$ and $\mu(1,\Delta,1)$ for $\Delta > \Delta_T$:
        \begin{align}
            \mu(1,\Delta,0) &\!=\! \frac{\theta}{p_t} \left[ \gamma_{11} \mathcal{X}(\Delta\!-\!1) \!+\!  \gamma_{01} \mathcal{X}(\Delta) \!+\! p_t \mu(1,\Delta\!-\!1,0) \right] \notag \\
            &=\frac{\beta \theta}{1-\gamma_{00}} \mathcal{X}(\Delta-1) + \theta \mu(1,\Delta\!-\!1,0), \label{eqn_Mu1D0_Diff_0} \\
            \mu(1,\Delta,1) &\!=\! \frac{q}{\bar{q}} \mu(1,\Delta,0).
        \end{align}
        
        We see that \eqref{eqn_Mu1D0_Diff_0} forms a first-order difference equation:
        \begin{align}
            \mu(1,\Delta,0) = \theta \mu(1,\Delta\!-\!1,0) + \mathcal{W}(\Delta), \quad \Delta>\Delta_T,
        \end{align}
        with the known initial value $\mu(1,\Delta_T,0)$, where we have defined $\mathcal{W}(\Delta) = \frac{\beta \theta}{1-\gamma_{00}} \mathcal{X}(\Delta-1)$ as a function of $\Delta$. The solution to this difference equation with constant coefficients is given by:
        \begin{align}
            \mu(1, \Delta, 0) \!=\! \theta^{\Delta - \Delta_T} \mu(1, \Delta_T, 0) \!+\! \sum_{k=\Delta_T+1}^{\Delta} \theta^{\Delta - k} \mathcal{W}(k),
        \end{align}
        where the summation can be obtained as:
        \begin{align}
            \sum_{k=\Delta_T+1}^{\Delta} \theta^{\Delta - k} \mathcal{W}(k) &= \frac{\beta \theta^\Delta \Psi}{(1-\gamma_{00})^2} \sum_{k=\Delta_T+1}^{\Delta} \left( \frac{\rho}{\theta}\right)^{k-1} \\
            &= \frac{\beta \rho^{\Delta_T} }{\gamma_{10}(1-\gamma_{00})} \frac{\theta^{\Delta-\Delta_T}-\rho^{\Delta-\Delta_T}}{\frac{1}{\rho} - \frac{1}{\theta}} \Psi. \notag
        \end{align}

        Thus,
        \begin{align}
            \mu(1,\Delta,0) &\!=\! \theta^{\Delta - \Delta_T} \mu(1, \Delta_T, 0)  \\
            & +  \frac{\beta \rho^{\Delta_T} }{\gamma_{10}(1-\gamma_{00})} \frac{\theta^{\Delta-\Delta_T}-\rho^{\Delta-\Delta_T}}{\frac{1}{\rho} - \frac{1}{\theta}} \Psi, \ \Delta>\Delta_T, \notag\\
            \mu(1,\Delta,1) &\!=\! \frac{q}{\bar{q}} \mu(1,\Delta,0), \quad \Delta>\Delta_T.
        \end{align}

        These last two equations also hold for $\Delta = \Delta_T$. The stationary distribution $\mu(b, \Delta, r)$ was obtained for the case $\Delta_T \geq 2$. Similarly, for $\Delta_T = 1$ and $\Delta_T = 0$, the distribution can be derived by solving the modified equations $\mathcal{E}'_5$, $\mathcal{E}''_4$, and $\mathcal{E}''_5$.
        Note that for $\Delta_T = 1$, equations $\mathcal{E}_1$ through $\mathcal{E}_4$ and $\mathcal{E}_8$ remain valid, while $\mathcal{E}_6$ and $\mathcal{E}_7$ are omitted. For $\Delta_T = 0$, equations $\mathcal{E}_1$ through $\mathcal{E}_3$ remain valid, $\mathcal{E}_6$ and $\mathcal{E}_7$ are omitted, and $\mathcal{E}_8$ remains valid for $\Delta \geq 1$ as follows:
        \begin{align}
            \mu(1,\Delta,0) &\!=\! \theta^{\Delta - 1} \mu(1, 1, 0)  \\
            & +  \! \frac{\beta \rho }{\gamma_{10}(1\!-\!\gamma_{00})} \frac{\theta^{\Delta-1}\!-\!\rho^{\Delta-1}}{\frac{1}{\rho} - \frac{1}{\theta}} \Psi, \ \ \Delta\geq 1, \ \Delta_T=0, \notag\\
            \mu(1,\Delta,1) &\!=\! \frac{q}{\bar{q}} \mu(1,\Delta,0), \quad \Delta\geq 1, \ \Delta_T=0.
        \end{align}

        The system of equations $\mathcal{E}'_5$ yields:
        \begin{align}
            \label{eqn_Mu110_DT1}
            \mu(1,1,0) =  \theta \left[ \frac{1}{p_t} - \frac{\rho}{1-\gamma_{00}} \right]\Psi, \quad \Delta_T=1, \\
            \label{eqn_Mu111_DT1}
            \mu(1,1,1) = \frac{q}{\bar{q}} \theta \left[ \frac{1}{p_t} - \frac{\rho}{1-\gamma_{00}} \right]\Psi, \quad \Delta_T=1.
        \end{align}

        The systems of equations $\mathcal{E}''_4$ and $\mathcal{E}''_5$ yield:
        \begin{align}
            \label{eqn_Mu100_DT0}
            \mu(1,0,0) = \frac{\gamma_{01} \theta}{p_t \left(1 - \gamma_{00} \right)} \Psi,, \quad \Delta_T=0, \\
            \label{eqn_Mu101_DT0}
            \mu(1,0,1) = \frac{q}{\bar{q}} \frac{\gamma_{01} \theta}{p_t \left(1 - \gamma_{00} \right)} \Psi,, \quad \Delta_T=0.
        \end{align}

        \begin{align}
            \label{eqn_Mu110_DT0}
            \mu(1,1,0) =  \frac{ \beta \theta \left( \rho + \gamma_{00} \theta \right) }{\gamma_{10} (1 - \gamma_{00})}  \Psi, \quad \Delta_T=0, \\
            \label{eqn_Mu111_DT0}
            \mu(1,1,1) = \frac{q}{\bar{q}} \frac{ \beta \theta \left( \rho + \gamma_{00} \theta \right) }{\gamma_{10} (1 - \gamma_{00})}  \Psi, \quad \Delta_T=0.
        \end{align}
        
        We can now solve equations \eqref{eqn_PsiEqnDTg2}, \eqref{eqn_PsiEqnDT1}, and \eqref{eqn_PsiEqnDT0} to derive $\Psi$ for the cases $\Delta_T \geq 2$, $\Delta_T = 1$, and $\Delta_T = 0$, respectively.

        \begin{align}
            \text{\eqref{eqn_PsiEqnDTg2}} &\Rightarrow \mu(0,0,:) \!+\!\!\sum_{\Delta=1}^{\infty} \!\! \mathcal{X}(\Delta) \!+\! \mu(1,0,:) \!+\!\!\sum_{\Delta=1}^{\Delta_T\!-\!1} \!\! \mathcal{Y}(\Delta) \!+\! \frac{\Psi}{q} \!=\! 1 \notag \\
            &\Rightarrow \frac{\bar{p}_t}{p_t} \Psi + \frac{(\Delta_T-1)\Psi}{p_t} + \sum_{\Delta=\Delta_T}^{\infty} \! \mathcal{X}(\Delta) +\frac{\Psi}{q} = 1,
        \end{align}
        where we have substituted $\mu(0,0,:) + \mu(1,0,:) = \frac{\bar{p}_t}{p_t} \Psi$ and $\mathcal{Y}(\Delta) = \frac{\Psi}{p_t} - \mathcal{X}(\Delta)$ for $1 \le \Delta \le \Delta_T-1$. The summation is also given by $\sum_{\Delta=\Delta_T}^{\infty} \mathcal{X}(\Delta) = \frac{\rho^{\Delta_T} \Psi}{(1-\gamma_{00})(1-\rho)} = \frac{\rho^{\Delta_T}}{\beta} \Psi$; therefore, $\Psi$ is derived as:
        \begin{align}
            \Psi = \frac{q p_t \beta}{q \left(\beta \Delta_T + p_t \rho^{\Delta_T}\right) + \bar{q} p_t \beta}. \label{eqn_PsiDTg2}
        \end{align}

    Similarly, \eqref{eqn_PsiEqnDT1} and \eqref{eqn_PsiEqnDT0} can be solved for $\Delta_T = 1$ and $\Delta_T = 0$, respectively.
    \begin{align}
        &\text{\eqref{eqn_PsiEqnDT1}} \Rightarrow \mu(0,0,:) \!+\!\sum_{\Delta=1}^{\infty} \! \mathcal{X}(\Delta) \!+\! \mu(1,0,:) \!+\! \frac{\Psi}{q} \!=\! 1  \notag \\
        &\Rightarrow \frac{\bar{p}_t}{p_t} \Psi \!+\! \frac{\rho}{\beta} \Psi \!+\! \frac{1}{q}\Psi \!=\! 1 
        \Rightarrow
        \Psi \!=\! \frac{q p_t \beta}{q \left(\beta \!+\! p_t \rho\right) \!+\! \bar{q} p_t \beta}, \ \ \Delta_T \!=\! 1. \label{eqn_PsiDT1} 
    \end{align}

    \begin{align}
        &\text{\eqref{eqn_PsiEqnDT0}} \Rightarrow \mu(0,0,:) +\sum_{\Delta=1}^{\infty}  \mathcal{X}(\Delta) + \frac{\Psi}{q} = 1 \notag \\
        &\Rightarrow \frac{\gamma_{00}}{1-\gamma_{00}} \Psi + \frac{\rho}{\beta} \Psi +\frac{1}{q}\Psi = 1 
        \Rightarrow
        \Psi = \frac{q \beta}{ q\bar{\beta}+ \beta}, \quad \Delta_T \!=\! 0. \label{eqn_PsiDT0}
    \end{align}

        The resulting $\Psi$ in \eqref{eqn_PsiDT1} and \eqref{eqn_PsiDT0} can be obtained from \eqref{eqn_PsiDTg2} by setting $\Delta_T$ to $1$ and $0$, respectively. Therefore, \eqref{eqn_PsiDTg2} holds for all $\Delta_T \in \{0, 1, 2, \dots\}$, which concludes the proof. 
    \end{proof}

    \section{Proof of Theorem \ref{Theorem_AvgQVAoI}}
    \label{Apen_Theorem_AvgQVAoI}

     The average QVAoI is obtained by \eqref{eqn_AvgQVAoI}. Substituting $\mathbb{P}(s^\prime \mid s) = \mathbb{P}(r^\prime) \mathbb{P}(b^\prime,\Delta^\prime \mid b,\Delta,r)$, we can further simplify it: 
    \begin{align}
        \Delta_{Avg}^{QVAoI} \!\!=\!\! \sum_{(b,\Delta)} \mu(b,\Delta,1) \! \! \sum_{(b^\prime,\Delta^\prime)} \! \Delta^\prime \ \mathbb{P}(b^\prime,\Delta^\prime \! \mid \! b,\Delta,1) .
    \end{align}

    The transition probabilities are given as follows: 
    \begin{align}
        \mathbb{P}(b', \Delta' \mid 0, \Delta, 1) = \begin{cases} 
        \bar{p}_t \bar{\beta}, & b'=0, \Delta'=\Delta, \\
        p_t \bar{\beta}, & b'=0, \Delta'=\Delta+1, \\
        \bar{p}_t \beta, & b'=1, \Delta'=\Delta, \\
        p_t \beta, & b'=1, \Delta'=\Delta+1. 
        \end{cases}
    \end{align}

    \begin{align}
        \Delta \!<\! \Delta_T\!: \ 
        \mathbb{P}(b', \Delta' \! \mid \! 1, \Delta, 1) \!=\! \begin{cases} 
        \bar{p}_t, & b'\!=\!1, \Delta'\!=\!\Delta, \\
        p_t, & b'\!=\!1, \Delta'\!=\!\Delta\!+\!1. 
        \end{cases}
    \end{align}

    \begin{align}
        \Delta \!\geq\! \Delta_T\!: \
        \mathbb{P}(b', \Delta' \! \mid \! 1, \Delta, 1) \!=\! \begin{cases} 
        \bar{p}_t \bar{\beta}, & b'\!=\!0, \Delta'\!=\!0, \\
        p_t \bar{\beta}, & b'\!=\!0, \Delta'\!=\!1, \\
        \bar{p}_t \beta, & b'\!=\!1, \Delta'\!=\!0, \\
        p_t \beta, & b'\!=\!1, \Delta'\!=\!1. 
        \end{cases}
    \end{align}

    The internal summation can be obtained using these transition probabilities:
    \begin{align}
        \Delta_{Avg}^{QVAoI} &= \sum_{\Delta=0}^{\infty} (\Delta + p_t) \mu(0,\Delta,1)  
        \\ 
        &+ \sum_{\Delta=0}^{\Delta_T-1} (\Delta + p_t) \mu(1,\Delta,1)  + \sum_{\Delta=\Delta_T}^{\infty} p_t \mu(1,\Delta,1). \notag \\
        &= \sum_{\Delta=1}^{\infty} \Delta \mu(0,\Delta,1) + \sum_{\Delta=1}^{\Delta_T-1} \Delta \mu(1,\Delta,1) \notag \\
        & + p_t \sum_{\Delta=0}^{\infty} \left[\mu(0,\Delta,1) + \mu(1,\Delta,1) \right] \notag \\
        &= \sum_{\Delta=1}^{\infty} q \Delta \mathcal{X}(\Delta) + \sum_{\Delta=1}^{\Delta_T-1} q \Delta \mathcal{Y}(\Delta) + q p_t, \notag 
    \end{align}
    which can be expressed as follows:
    \begin{align}
        \Delta&_{Avg}^{QVAoI} = 
            q \sum_{\Delta=\Delta_T}^{\infty} \Delta \mathcal{X}(\Delta) + \frac{q\Delta_T (\Delta_T-1)}{2 p_t} \Psi + qp_t,
    \end{align}
    where we have substituted $\mathcal{Y}(\Delta) = \frac{\Psi}{p_t} - \mathcal{X}(\Delta)$. The summation $\sum_{\Delta=\Delta_T}^{\infty} \Delta \mathcal{X}(\Delta) = \frac{\Psi}{1-\gamma_{00}} \sum_{\Delta=\Delta_T}^{\infty} \Delta \rho^\Delta$ can be written as:
    \begin{align}
        \sum_{\Delta=\Delta_T}^{\infty} \! \Delta \mathcal{X}(\Delta) \!=\! \frac{\rho^{\Delta_T} \Psi}{1\!-\!\gamma_{00}} \frac{\Delta_T \!-\! (\Delta_T \!-\! 1)\rho}{(1\!-\!\rho)^2} = \frac{\rho^{\Delta_T} \Psi}{1\!-\!\gamma_{00}}  \mathcal{G}_{\Delta_T}(\rho),
    \end{align}
    where we have defined $\mathcal{G}_{\Delta_T}(x)  \triangleq  \sum_{\Delta=\Delta_T}^{\infty} \Delta \times x^{\Delta-\Delta_T} = \frac{\Delta_T - (\Delta_T - 1)x}{(1-x)^2}$. The resulting average QVAoI is, therefore:
    \begin{align}
        \Delta_{Avg}^{QVAoI} 
             &\!=\! q p_t \!+\! \frac{q \Psi}{p_t \bar{\beta}} \left[ \bar{\beta} \frac{\Delta_T(\Delta_T\!-\!1)}{2} \!+\! \rho^{\Delta_T\!+\!1} \mathcal{G}_{\Delta_T}(\rho) \right]\!,
    \end{align}
    where we have substituted $\frac{p_t}{1-\gamma_{00}}=\frac{\rho}{\bar{\beta}}$.

    To derive the average VAoI, we first obtain the marginal distribution of VAoI, denoted by $\mu_\Delta(\delta)$, which is given by: $\mu_\Delta(\delta) = \sum_{b \in \{0,1\}} \sum_{r \in \{0,1\}} \mu(b,\delta,r)=\sum_{b \in \{0,1\}} \mu(b,\delta,:)$. For $\Delta_T \geq 1$:
    \begin{align}
        \mu_\Delta(0) &= \mu(0,0,:) + \mu(1,0,:) = \frac{\bar{p}_t}{p_t} \Psi, \\
        \mu_\Delta(\Delta) &= \mu(0,\Delta,:) + \mu(1,\Delta,:) = \frac{1}{p_t} \Psi, \quad 1 \leq \Delta < \Delta_T, \\
        \mu_\Delta(\Delta) &= \mu(0,\Delta,:) + \mu(1,\Delta,:) 
        = \mathcal{X}(\Delta) + \frac{1}{\bar{q}} \mu(1,\Delta,0) \notag \\
        &= \mathcal{X}(\Delta) + \frac{\theta^{\Delta-\Delta_T+1}}{\bar{q}} \mathcal{Y}(\Delta_T) \\
        &+ \frac{\beta \rho^{\Delta_T} }{\bar{q} \gamma_{10}(1-\gamma_{00})} \frac{\theta^{\Delta-\Delta_T}-\rho^{\Delta-\Delta_T}}{\frac{1}{\rho} - \frac{1}{\theta}} \Psi, \quad \Delta \geq \Delta_T.  \notag
    \end{align}

    Using thin marginal distribution, the average VAoI for $\Delta_T \geq 1$ can be derived as:
    \begin{align}
        \Delta_{Avg}^{VAoI} &=  \! \sum_{\Delta=0}^{\infty} \Delta \mu_\Delta(\Delta) = \Biggl[ \frac{\Delta_T(\Delta_T-1)}{2 p_t} \\
        &+ \frac{\rho^{\Delta_T} \mathcal{G}_{\Delta_T}(\rho)}{1-\gamma_{00}} 
        + \! \frac{\theta \mathcal{G}_{\Delta_T}(\theta)}{\bar{q}} \left( \frac{1}{p_t} - \frac{\rho^{\Delta_T}}{1-\gamma_{00}} \right) \notag \\
        & +\! \frac{\beta \rho^{\Delta_T}}{\bar{q} \gamma_{10} (1\!-\!\gamma_{00}) (\frac{1}{\rho} \!-\! \frac{1}{\theta})} \Big( \mathcal{G}_{\Delta_T}(\theta) \!-\! \mathcal{G}_{\Delta_T}(\rho) \Big) \Biggr] \Psi. \notag
    \end{align}

    For $\Delta_T = 0$:
    \begin{align}
        \mu_\Delta(0) &= \mu(0,0,:) \!+\! \mu(1,0,:) \!=\! \frac{\bar{p}_t (1 \!-\! \bar{q} \bar{p}_t \bar{\beta})}{(1\!-\!\bar{q} \bar{p}_t)(1 \!-\! \gamma_{00})} \Psi, \\
        \mu_\Delta(\Delta) &= \mu(0,\Delta,:) \!+\! \mu(1,\Delta,:) 
        \!=\! \mathcal{X}(\Delta) \!+\! \frac{1}{\bar{q}} \mu(1,\Delta,0), \notag \\
        &= \mathcal{X}(\Delta) + \frac{\theta^{\Delta-1}}{\bar{q}} \frac{ \beta \theta \left( \rho + \gamma_{00} \theta \right) }{\gamma_{10} (1 - \gamma_{00})}  \Psi \notag \\
        &+ \frac{\beta \rho }{\bar{q} \gamma_{10}(1\!-\!\gamma_{00})} \frac{\theta^{\Delta-1}\!-\!\rho^{\Delta-1}}{\frac{1}{\rho} - \frac{1}{\theta}} \Psi \notag \\
        &= \mathcal{X}(\Delta) + \frac{\beta \bar{\beta} \theta}{\bar{q}^2} \frac{\theta^{\Delta+1}-\frac{\bar{q}}{\bar{\beta}} \rho^{\Delta+1}}{\gamma_{10} (1-\gamma_{00}) (\theta-\rho) } \Psi, \quad \Delta \geq 1.
    \end{align}

    The average VAoI for $\Delta_T = 0$ can be derived as:
    \begin{align}
        \Delta_{Avg}^{VAoI} & = \left[ \frac{\rho}{(1-\gamma_{00})(1-\rho)^2} 
        +  \frac{\beta \bar{\beta} \theta}{\bar{q}^2} \frac{\frac{\theta^{2}}{(1-\theta)^2}-\frac{\bar{q}}{\bar{\beta}} \frac{\rho^{2}}{(1-\rho)^2}}{\gamma_{10} (1-\gamma_{00}) (\theta-\rho) } \right] \Psi.
    \end{align}

\end{document}